\documentclass[runningheads]{llncs}

\usepackage[utf8]{inputenc}
\usepackage[T1]{fontenc}
\usepackage[USenglish]{babel}
\usepackage[intlimits]{amsmath}
\usepackage{amsfonts}
\usepackage{amssymb}
\usepackage{amsmath}

\usepackage{mathtools}
\usepackage{setspace}
\usepackage{graphicx}
\usepackage{hyperref}

\usepackage[dvipsnames]{xcolor}
\usepackage{todonotes}
\usepackage{tikz}
\usetikzlibrary{shapes,decorations,shapes.geometric,backgrounds,fit,positioning}
\usepackage{algorithm}
\usepackage{comment}
\usepackage[noend]{algpseudocode}
\newcounter{observ}
\newtheorem{observation}[observ]{Observation}
\newtheorem{brarule}{Branching Rule}
\newtheorem{redrule}{Reduction Rule}

\usepackage{subcaption}

\newcommand{\Oh}{\mathcal{O}}

\usepackage{xspace}

\newcommand\NP{\ensuremath{\textsf{NP}}\xspace}

\begin{document}

\title{Enumerating Connected Dominating Sets}
\author{Faisal N. Abu-Khzam\inst1\orcidID{0000-0001-5221-8421} 
\and Henning Fernau\inst2\orcidID{0000-0002-4444-3220}
\and Benjamin Gras
\inst3\orcidID{0000-0003-3989-063X}
\and Mathieu Liedloff\inst3\orcidID{
0000-0003-2518-606X}
\and Kevin Mann\inst2\orcidID{0000-0002-0880-2513} 
}
\authorrunning{F. Abu-Khzam, H. Fernau, B. Gras, M. Liedloff and K. Mann.}
\institute{
Department of Computer Science and Mathematics\\
Lebanese American University, 
Beirut, Lebanon.\\
\email{faisal.abukhzam@lau.edu.lb}\\
\and
Universit\"at Trier, Fachbereich~4 -- Abteilung Informatikwissenschaften\\  
54286 Trier, Germany.\\
\email{\{fernau,mann\}@uni-trier.de}
\and
Université d'Orléans, INSA Centre Val de Loire, LIFO EA 4022, Orléans, France.\\
\email{\{bejamin.gras,mathieu.liedloff\}@univ-orleans.fr}
}
\maketitle 


\newcommand{\cdsbasis}{1.9767}
\newcommand{\genbasis}{1.9896}
\newcommand{\avalue}{0.106}
\newcommand{\dvalue}{0.106}

\newcommand{\oldcdslowerbound}{1.4422}
\newcommand{\newcdslowerbound}{1.4890}
\newcommand{\newcdsthreedegenlowerbound}{1.4723}
\newcommand{\newcdstwodegenlowerbound}{1.3195}

\begin{abstract}
The question to enumerate all (inclusion-wise) minimal connected dominating sets in a graph of order~$n$ in time significantly less than $2^n$ is an open question that was asked in many places.
We answer this question affirmatively, by providing an enumeration algorithm that runs in time $\Oh(\genbasis^n)$, using polynomial space only. The key to this result is the consideration of this enumeration problem on 2-degenerate graphs, which is proven to be possible in time $\Oh(\cdsbasis^n)$. Apart from solving this old open question, we also show new lower bound results. More precisely, we construct a family of graphs of order~$n$ with $\Omega(\newcdslowerbound^n)$ many minimal connected dominating sets, while previous examples achieved 
$\Omega(\oldcdslowerbound^n)$.
Our example happens to yield 4-degenerate graphs. 
Additionally, we give lower bounds for the previously not considered classes of 2-degenerate and of 3-degenerate graphs, which are $\Omega(\newcdstwodegenlowerbound^n)$ and
$\Omega(\newcdsthreedegenlowerbound^n)$, respectively. 

We also address essential questions concerning output-sensitive enumeration. Namely, we
give reasons why our algorithm cannot be turned into an enumeration algorithm that guarantees polynomial delay without much efforts. More precisely, we prove that it is \NP-complete to decide, given a graph~$G$ and a vertex set~$U$, if there exists a minimal connected dominating set~$D$ with $U\subseteq D$, even if $G$ is known to be 2-degenerate. Our reduction also shows that even any subexponential delay is not easy to achieve for enumerating minimal connected dominating sets. 
Another reduction shows that no \textsf{FPT}-algorithms can be expected for this extension problem concerning minimal connected dominating sets, parameterized by $|U|$. 
This also adds one more problem to the still rather few natural parameterized problems that are complete for the class \textsf{W}[3].
We also relate our enumeration problem  to the famous open \textsc{Hitting Set Transversal}  problem, which can be phrased in our context as the question to enumerate all minimal dominating sets of a graph with polynomial delay by showing that a polynomial-delay enumeration algorithm for minimal connected dominating sets implies an
affirmative algorithmic solution to the \textsc{Hitting Set Transversal}  problem.
\end{abstract}

\newpage

\section{Introduction}

The enumeration of objects that satisfy a given property has applications in many scientific domains including biology and artificial intelligence. Enumeration can also be used as part of an exact algorithm, \emph{e.g.,} confer the algorithm by Lawler \cite{Law76} to compute a coloring of an input graph using a minimum number of colors. The dynamic programming scheme used by this algorithm needs all the maximal independent sets of the input graph. It is worth noting that the running time depends mainly on a bound on the number of maximal independent sets as well as on the running time of an algorithm that would produce all these sets.

Clearly, the number of outputs of an enumeration algorithm can be exponential in the size of the given input. 
It is the case for the number of maximal independent sets: there are graphs with $3^{n/3}$ such sets \cite{MooMos65}, where $n$ is the number of vertices in the graph. The running time of enumeration algorithms can either be measured with respect to the size of the input plus the size of the outputted set of objects, which is called \emph{output-sensitive} analysis, or it can be measured according to the size of the input only, being called \emph{input-sensitive} analysis. In the latter, the running time upper bound often implies an upper bound on the number of enumerated objects, \emph{i.e.}, the maximum number of objects that can fulfill the given property. 

Given a graph $G=(V,E)$, the problem of computing a minimum dominating set asks for a smallest-cardinality subset $S\subseteq V$ such that each vertex not in $S$ has at least one neighbor in~$S$. This well studied \NP-hard problem attracted considerable attention for decades. Several exponential-time algorithms have been designed to solve the problem exactly, and the most recent are based on \emph{Measure-and-Conquer} techniques to analyze their running times \cite{FomGraKra2009,NedRooDij2014,Iwa1112}. The problem of enumerating all inclusion-minimal dominating sets has also caught attention for general graphs as well as for special graph classes \cite{Fometal2008a,CouLetLie2015,GolHKKV2017}.

Many variants of the dominating set problem have also gained attention~\cite{HHS98}. In particular, the minimum \emph{connected} dominating set problem requires that the graph induced by~$S$ be connected.
The problem has attracted great attention and various methods have been devised to solve it exactly \cite{AbuMouLie2011,FomGraKra2008}. A more challenging question has been posed about the enumeration of inclusion-minimal {connected} dominating sets. Already designing an algorithm faster than $\text{poly}(n) 2^n$ is known to be challenging, and this specific question has been asked several times as an open problem \cite{Utrecht-TR-2015-016,FerGolSag2018,GolHegKra2016}. 
A recent result by Lokshtanov \emph{et al.}~\cite{LokPilSau2018} shows that minimal {connected} dominating sets can be enumerated in time $2^{(1-\epsilon)n}\cdot n^{\mathcal{O}(1)}$, which broke the $2^n$-barrier for the first time. It is worth noting that $\epsilon$ is a tiny constant, around $10^{-50}$, and it has remained open whether an algorithm exists that can substantially break the $2^n$-barrier. The enumeration of minimal {connected} dominating sets also received notable interest when the input is restricted to special graph classes \cite{GolHegKra2016,GolHKS2020,Skj2017,Say2019,Say2019a}. 

On the other hand, the maximum number of minimal CDS in a graph was shown to be in $\Omega(3^{\frac{n}{3}})$ \cite{GolHegKra2016}, which is obviously very low compared to the currently best upper bound. This gap between upper and lower bounds is narrower when it comes to special graph classes. On chordal graphs, for example, the upper bound has been recently improved to $\Oh(1.4736^n)$ \cite{GolHKS2020}. Other improved lower/upper bounds have been obtained for AT-free, strongly chordal, distance-hereditary graphs, and cographs in \cite{GolHegKra2016}. Further improved bounds for split graphs, cobipartite and convex bipartite graphs have been obtained in \cite{Skj2017} and 
\cite{Say2019}. Moreover, although the optimization problem seems simpler, the best-known exact algorithm solves the problem in time $O(1.8619^n)$ \cite{AbuMouLie2011}. This is already much larger than the best-known lower bounds of $3^{(n-2)/3}$ \cite{GolHegKra2016} to enumerate all minimal connected dominated sets.

In this paper, we show that the enumeration of all inclusion-wise minimal connected dominating sets can be achieved in time $\Oh(\genbasis^n)$. Surprisingly, achieving this improvement was simply based on first considering the same enumeration on 2-degenerate graphs and proving it to be possible in time $\Oh(\cdsbasis^n)$. Achieving enumeration with polynomial delay is believed to be hard, since it would also lead to the same for the enumeration of minimal dominating sets, which has been open for several decades. We give further evidence of this (possible) hardness by showing that extending a subset of vertices into a minimal connected dominating set is \NP-complete and also hard in a natural parameterized setting. Furthermore, we narrow the gap between upper and lower bounds by showing that the maximum number of minimal connected dominating sets in a graph is in $\Omega(\newcdslowerbound^n)$, thus improving the previous lower bound of $\Omega(\oldcdslowerbound^n)$. Our construction yields new lower bounds on several special graph classes such as 3-degenerate planar bipartite graphs.

\section{Definitions, Preliminaries and Summary of Main Results}

In this paper, we deal with undirected simple finite graphs that can be specified as $G=(V,E)$, where $V$ is the finite vertex set and $E\subseteq\binom{V}{2}$ is the set of edges. The number of vertices $|V|$ is also called the \emph{order} of graph~$G$ and is denoted by $n$.
An edge $\{u,v\}$ is usually written as $uv$.  Alternatively, $E$ can be viewed as a symmetric binary relation, so that $E^*$ is then the transitive closure of $E$, which is an equivalence relation whose equivalence classes are also known as \emph{connected components}.
A graph is called \emph{connected} if it has only one connected component. A graph  $G'=(V',E')$ is a \emph{subgraph} of $G=(V,E)$ if $V'\subseteq V$ and $E'\subseteq E$; $G'$ is a \emph{partial graph} of~$G$ if $V=V'$.
A set of vertices~$S$ \emph{induces} the subgraph $G[S]=(S,E_S)$, where $E_S=\{uv\in E\mid u,v\in S\}$;
$S$ is called \emph{connected} if $G[S]$ is connected.
For a vertex $v\in V$, $N_G(v)=\{u\in V\mid uv\in E\}$ is the \emph{open neighborhood} of~$v$, collecting the vertices \emph{adjacent} to~$v$; its cardinality $|N_G(v)|$ is also called the \emph{degree} of~$v$, denoted as $\deg_G(v)$. We denote the \emph{closed neighborhood} of~$v$ by $N_G[v]=N_G(v)\cup\{v\}$. We can extend set-valued functions to set arguments; for instance,  $N_G[S]=\bigcup_{v\in S}N_G[v]$ for a set of vertices~$S$; $S$ is a \emph{dominating set} if $N_G[S]=V$. Whenever clear from context, we may drop the subscript $G$ from our notation. If $X\subseteq V$, we also write $N_X(v)$ instead of $N(v)\cap X$ and $\deg_X(v)$ for 
$|N(v)\cap X|$. For brevity, we write CDS for \emph{connected dominating set}.
Next, we collect some observations.

\begin{observation}\label{obs-CDS-on-partialgraphs}
If $S$ is a CDS of a partial graph $G'$ of $G$, then $S$ is a CDS of~$G$.
\end{observation}

\begin{proof}
We can think of $G=(V,E)$ as being obtained from $G'$ by adding edges. Hence, if $N_{G'}[S]=V$, then $N_G[S]=V$. Moreover, adding edges cannot violate connectivity.
\qed \end{proof}

\begin{corollary}\label{cor-minCDS-on-partialgraphs}
Let  $S$ be a CDS both of $G$ and of a partial graph $G'$ of $G$. If $S$ is a minimal CDS of $G$, then it is a minimal CDS of~$G'$.
\end{corollary}

\begin{proof}
Consider such a CDS $S$ on $G$ and $G'$.
If $X\subsetneq S$ is a CDS on $G'$, then it is also a CDS on $G$ by
Observation~\ref{obs-CDS-on-partialgraphs}.
Hence, if $S$ is a minimal CDS on $G$, then it is also a minimal CDS on the partial graph~$G'$. 
\qed \end{proof}

A graph $G=(V,E)$ is \emph{$d$-degenerate} if there exists an \emph{elimination ordering} $(v_1,\dots,v_n)$, where $V=\{v_1,\dots,v_n\}$, such that $$\forall i=1,\dots,n:\deg_{G[\{v_i,\dots,v_n\}]}(v_i)\leq d\,.$$ In other words, we can subsequently delete $v_1,v_2,\dots$ from $G$, and at the time when $v_i$ is deleted, it has degree bounded by $d$ in the remaining graph. The decision problem \textsc{Connected Dominating Set Extension} expects as inputs a graph $G=(V,E)$ and a vertex set~$U$, and the question is if there exists a minimal CDS~$S$ that extends~$U$, \emph{i.e.}, for which $S\supseteq U$ holds.

In the next section, we develop a branching algorithm. It is classical to analyze its running-time by solving recurrences of type $T(\mu(\mathcal{I}) ) = \sum_{i=1}^{t} T(\mu(\mathcal{I})-r_i)$. Here, $\mu(\mathcal{I})$ is a measure on the size of the instance. The value of $t$ is the number of recursive calls ($t$ is equal to $1$ for \emph{reduction rules}) and each $r_i$ is (a lower bound on) the reduction of the measure corresponding to the recursive call. We simply denote by $(r_1, r_2, \dots, r_k)$ the \emph{branching vector} of the recurrence. We refer to the book by Fomin and Kratsch for further details on this standard analysis~\cite{FomKra2010}.

\medskip

As discussed in the introduction, we shall first prove that all minimal CDS can be enumerated in time $\mathcal{O}(\cdsbasis^n)$ on 2-degenerate graphs. This result is the key to our enumeration result for general graphs.
This is why we will first present the corresponding branching enumeration algorithm for 2-degenerate graphs in a simplified form and analyze it with a rather simple measure in order to explain its main ingredients, and only thereafter, 
we turn towards a refined analysis that finally leads to the claimed enumeration result on general graphs.
We shall prove that our input-sensitive enumeration algorithm cannot be turned into an enumeration algorithm with polynomial delay by simply interleaving the branching with tests for extendibility.
More precisely, we consider the following decision problem \textsc{Connected Dominating Set Extension}: instances are pairs $(G,U)$, where $G=(V,E)$ is a graph and $U\subseteq V$. The question is if there exists a minimal CDS $D$ with $U\subseteq D$ in~$G$.  
For possible applications of such extension algorithms, we refer to the discussions in~\cite{CasFGMS2022}.

\section{A CDS enumeration algorithm for 2-degenerate graphs}
\label{sec-simpleCDSalgo}

We are going to present an algorithm that enumerates all inclusion-wise minimal connected dominating sets (CDS) of a 2-degenerate graph $G=(V,E)$. Based on $G$, in the course of our algorithm, an instance is specified by 
$\mathcal{I}=\left(V';O_d,O_n;S\right)\,,$ consisting of four vertex sets that partition~$V$.
In the beginning, $\mathcal{I}=(V;\emptyset;\emptyset;\emptyset)$.
In general, $V'$ collects the vertices not yet decided by branching or reduction rules,
$O:=O_d\cup O_n$ are the vertices that have been decided \emph{not} to be put into the solution, while~$S$ is the set of vertices decided to go into the solution that is constructed by the branching algorithm. The set~$O$ is further refined into~$O_n$, the set of vertices that are not yet dominated,
\emph{i.e.}, if $x\in O_n$, then $\deg_S(x)=0$, and $O_d$, the set of vertices that are already dominated, \emph{i.e.}, if $x\in O_d$, then $\deg_S(x)>0$. Similarly, we will sometimes refine $V'=V'_d\cup V'_n$. 
In the leaves of the branching tree, only instances of the form $(\emptyset,O_d,\emptyset,S)$ are of interest. Yet, before outputting $S$ as a solution, one has to check if $S$ is connected and if it does not contain a smaller CDS.

The algorithm actually starts by creating $n$ different branches, in each a single vertex is put in $S$ as a starting point, so that $S$ is never empty. More precisely, the $n^\text{th}$ branch would decide \emph{not} to put the previously considered $n-1$ vertices into the solution but the $n^\text{th}$ one is put into~$S$. This binary branching avoids generating solutions twice. Also, it is trivial to check in each branch if the selected vertex already dominates the whole graph, so that we can henceforth assume that $G$ cannot be dominated by a single vertex.

\begin{remark}In the case of 2-degenerate graphs, we can also offer a good combinatorial understanding of graphs that are dominated by a single vertex. Namely, if $G$ is dominated by one single vertex $v$, then $G-v$ is $1$-degenerate, \emph{i.e.}, a forest. If $G-v$ has more than one connected component, then $v$ is the only minimal CDS of~$G$. If $G-v$ is a tree, \emph{i.e.}, it has one connected component, then deleting all leaves from this tree gives another minimal CDS of~$G$, but there are no other minimal CDS in~$G$.
\end{remark}

We denote by $c$ the number of connected components of $G[S]$.
Now, we are ready to define the measure that we use to analyze the running time of our algorithm, following a very simple version of  the \emph{measure-and-conquer}-paradigm, as explained in~\cite{FomGraKra2009,FomKra2010},  
$$\mu(\mathcal{I})=|V'|+ \alpha\cdot |O_n|+\delta \cdot c\,.$$

We decide that $0<\alpha,\delta<1$, but we will determine the concrete values later as to minimize the upper-bound on the running-time. 
At the beginning, $V'=V$, $O_n=\emptyset$ and $c=0$, so that then the measure equals~$|V|$.
At the end, $V'=O_n=\emptyset$ and the measure equals~$\delta$ if the solution is connected and is bigger than $\delta$ if the solution is not connected. 

For the possible branchings, we only consider vertices in a partial graph $G'$ of $G[V'\cup O_n]$. As~$G$ is 2-degenerate, $G'$ is also 2-degenerate, so that we can always find a vertex of degree at most two in~$G'$. Some of our branchings apply to vertices of arbitrary degree, though; in such a situation, we denote the vertex that we branch on as~$x$. If we branch on a small-degree vertex (due to 2-degeneracy), this vertex is called~$u$. Clearly, a binary branch that puts a selected vertex either into~$S$ or into~$O$ is a complete case distinction.

We are now explaining the conventions that we follow in our illustrations of subgraphs of~$G'$.
Vertices in~$V'$ are depicted by~\tikz[fill lower half/.style={path picture={\fill[#1] (path picture bounding box.south west) rectangle (path picture bounding box.east);}}]{\draw (0,0) node[circle,fill lower half=black!40,scale=0.8, draw] {};}
and more specifically by~\tikz{\draw[circle,fill=white, scale=0.8]  circle(1ex);} if in~$V'_n$
or by~\tikz{\draw[circle,fill=black!40, scale=0.8]  circle(1ex);} if in~$V'_d$. 
We use black squares~\tikz{\draw[rectangle,fill=black, scale=0.8] (0,0) rectangle (2ex,2ex);}
to depict vertices which are already decided to belong to the solution~$S$\footnote{It should be clear that one could always move to the graph where vertices in $S$ that belong to the same connected component in $G[S]$ are merged. In order to avoid drawing too many vertices from~$S$ in our pictures, we assume these mergings to have been performed, so that (in particular) when we draw two vertices from~$S$, they belong to different connected components.}. 
We use~\tikz{\draw[rectangle,fill=white,scale=0.8] (0,0) rectangle (2ex,2ex);} for vertices from~$O_n$.
So, circles are used for undecided vertices (these vertices might still be added to~$S$), whereas squares are used for vertices being already decided (to belong to the solution or to be discarded).
 Vertices in~$V'\cup O$ are depicted as half-filled diamonds~\tikz[fill lower half/.style={path picture={\fill[#1] (path picture bounding box.south west) rectangle (path picture bounding box.east);}}]{\draw (0,0) node[diamond,fill lower half=black!40,scale=0.7, draw] {};},
and if  the vertex is from $V'_n\cup O_n$, then we use an unfilled diamond~\tikz{\draw (0,0) node[diamond,scale=0.7, draw] {};} to represent it. A dashed line indicates an edge that may be present.

\medskip

In the following, the branching and reduction rules require to be executed in order, so that our instance will (automatically) satisfy some structural properties when we apply one of the later rules. We can separate our branching and reduction rules into three parts:
\begin{itemize}
    \item A first set of rules (Branching Rules~\ref{branching:dom-nbs-On}, \ref{branching:dom-diff-compos}, \ref{branching:dom-many-udo-nbs-with-one-On-nb}, \ref{branching:dom-nbs-diff-compos}, Reduction Rules~\ref{reduction:dom-2nonadjac-nbs-one-dom}, \ref{reduction:isolates}, \ref{reduction:edge-On}, \ref{reduction:udo-only-one-dom}) that deals with vertices of arbitrary degree that are (possibly) dominated, but only in some special cases. Those rules are applied first, so that whenever we apply a rule from the next two sets, we know that any dominated vertex is dominated by vertices from exactly one connected component of $G[S]$ and none of the vertices in its neighborhood are dominated, i.e., the set of vertices $V'_d$ as well as the set of vertices $O_n$ forms an independent set in $G'$, a partial graph of $G[V'\cup O_n]$. 
    \item A second set of rules (Branching Rules~\ref{branching:udo-2-no-On}, \ref{branching:On-2-no-On}, Reduction Rules~\ref{reduction:On-encircled}) that handles the cases where the small-degree vertex that exists by the $2$-degeneracy is undominated. If we apply a rule from the third set,  every undominated vertex is of degree at least~$3$. 
    \item The last set of rules (Branching Rules~\ref{branching:dom-one-udo-nb-with-no-On-nbs}, \ref{branching:dom-2nbs-udo-common-nb},  \ref{branching:dom-2nbs-udo-nb-y}, \ref{branching:dom-2nbs-with-many-nbs}) handles only the cases where the small-degree vertex that exists by the $2$-degeneracy is dominated.
\end{itemize}
Note that even inside those three sets, rules have to be executed in the given order.

\begin{brarule}\label{branching:dom-nbs-On}
Let $x\in V'_d$  with $\deg_{O_n}(x)\geq 1$ (\autoref{fig:branching:dom-nbs-On}). 
Then branch as follows.
\begin{enumerate}
    \item Put $x$ in $O_d$.
    \item Put $x$ in $S$ and every vertex in $N_{O_n}(x)$ in $O_d$.
\end{enumerate}
\end{brarule}

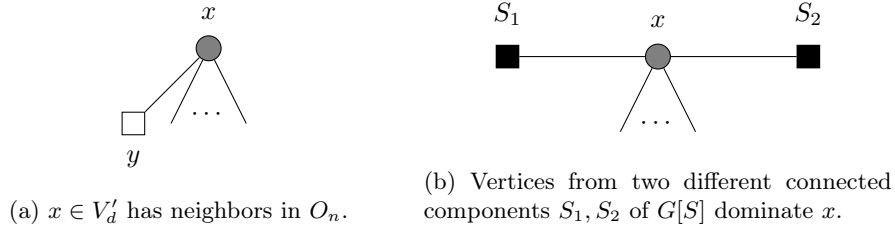
\begin{figure}[bt]
    \centering
    	
\begin{subfigure}[b]{.4\textwidth}
    \centering
    	
	\begin{tikzpicture}[transform shape]
			\tikzset{every node/.style={ fill = black,circle,minimum size=0.3cm}}
			\node[draw,fill=gray,label={above:$x$}] (u) at (7,0) {};
			\node[draw,rectangle,fill=white,label={below:$y$}] (v1) at (6,-1) {};
			\node[fill=none,label={below:$\cdots$}] at (7,-0.3) {};
			\path (u) edge[-] (v1);
			\path (u) edge (6.5,-1);
			\path (u) edge (7.5,-1);
        \end{tikzpicture}

    \subcaption{$x\in V'_d$ has neighbors in $O_n$.}
    \label{fig:branching:dom-nbs-On}
\end{subfigure}
\qquad
\begin{subfigure}[b]{.51\textwidth}
    \centering
    	
	\begin{tikzpicture}[transform shape]
			\tikzset{every node/.style={ fill = black,circle,minimum size=0.3cm}}
			\node[draw,fill=gray,label={above:$x$}] (x) at (7,0) {};
			\node[draw,rectangle,fill=black,label={above:$S_1$}] (s1) at (5,0) {};
			\node[draw,rectangle,fill=black,label={above:$S_2$}] (s2) at (9,0) {};
			\node[fill=none,label={below:$\cdots$}] at (7,-0.3) {};
			\path (x) edge[-] (s1);
			\path (x) edge[-] (s2);
			\path (x) edge (6.5,-1);
			\path (x) edge (7.5,-1);
			
        \end{tikzpicture}

    \subcaption{Vertices from two different connected components $S_1,S_2$ of $G[S]$ dominate~$x$.}
    \label{fig:branching:dom-diff-compos}
    \end{subfigure}
    \caption{Simple branchings for dominated vertices: Rules~\ref{branching:dom-nbs-On}  and~\ref{branching:dom-diff-compos}.}
\end{figure}

\begin{lemma}
The branching of rule~\ref{branching:dom-nbs-On} is a complete case distinction. Moreover, it leads to a branching vector that is not worse than
$(1,1+\alpha)\,.$
\end{lemma}

\begin{proof} 
As $\deg_{S}(x)\geq 1$, we  find that, if $x$ is not put into the solution, then it is put into $O_d$, which decreases the measure by~1 in the first component of the branching vector.  If~$x$ is put into the solution, then its neighbors are dominated and we decrease the measure by at least $1+ \alpha$ in the second component  of the branching vector, as $\deg_{O_n}(x)\geq 1$. 
\qed \end{proof}

\begin{brarule}
\label{branching:dom-diff-compos}
Let $x\in V'_d$ such that $x$ is adjacent to two different connected components of $G[S]$; see \autoref{fig:branching:dom-diff-compos}. 
Then branch as follows.
\begin{enumerate}
    \item Put $x$ in $O_d$.
    \item Put $x$ in $S$. 
\end{enumerate}
\end{brarule}

\begin{lemma}
The branching of rule~\ref{branching:dom-diff-compos} is a complete case distinction. Moreover, it leads to a branching vector that is not worse than
$(1,1+\delta)\,.$
\end{lemma}

\begin{proof}
When $x$ is put into~$O_d$, the measure drops by one.
In the second branch, the number of connected components of $G[S]$ decreases by $1$, so the measure decreases by $1+\delta$ in total in this case.  
\qed \end{proof}

We are now presenting two branching rules that could be viewed as variations of the first two; they always give worse branchings.

\begin{brarule}\label{branching:dom-many-udo-nbs-with-one-On-nb}
Let $x\in V'_d, y\in N_{V'_n}(x), z\in N_{O_n}(y)$ (\autoref{fig:branching:dom-many-udo-nbs-with-one-On-nb}). 
Then branch as follows.
\begin{enumerate}
    \item Put $x$ in $O_d$.
    \item Put $x$ in $S$, $y$ in $O_d$.
        \item Put $x,y$ in $S$ and thus $z \in O_d$.
\end{enumerate}
\end{brarule}

\begin{lemma}
The branching of rule~\ref{branching:dom-many-udo-nbs-with-one-On-nb} is a complete case distinction. Moreover, it leads to a branching vector that is not worse than
$(1,2,2+\alpha)\,.$
\end{lemma}

\begin{proof}
Let $M$ be a minimal CDS of $G$ such that $M \setminus V' =S$. Either we put $x$ in $O_d$ (as it is already dominated, this decreases the measure by~1) or in~$S$. Assume $x\in S$. Then $y$ would be dominated by~$x$. Therefore, $y$ is either in $O_d$ (decreasing the measure by~2) or in~$S$. For $y\in S$, $z$ is dominated. Thus, $z\in O_d$ has to hold and the measure is decreased by $2+\alpha$.   
\qed \end{proof}

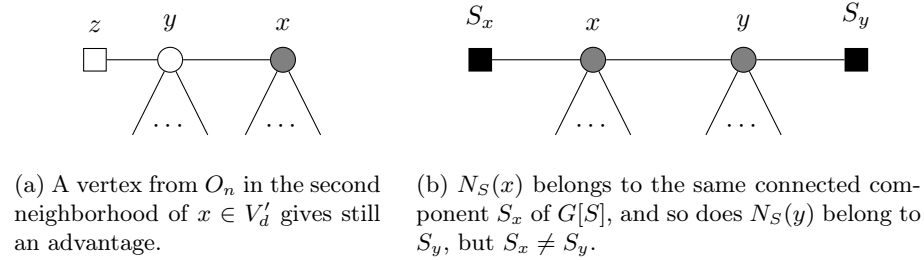
\begin{figure}
    \centering
    	
\begin{subfigure}[b]{.4\textwidth}
    \centering
    	
	\begin{tikzpicture}[transform shape,fill lower half/.style={path picture={\fill[#1] (path picture bounding box.south west) rectangle (path picture bounding box.east);}}]
			\tikzset{every node/.style={ fill = black,circle,minimum size=0.3cm}}
			\node[draw,fill=gray,label={above:$x$}] (u) at (7.5,0) {};
			\node[draw,fill=white,label={above:$y$}] (v1) at (6,0) {};
			\node[draw,fill=white,rectangle,label={above:$z$}] (y) at (5,0) {};
			\node[fill=none,label={below:$\cdots$}] at (6,-0.3) {}; 
			\node[fill=none,label={below:$\cdots$}] at (7.5,-0.3) {};
			\path (u) edge[-] (v1);
			\path (v1) edge[-] (y);
			\path (v1) edge (5.5,-1);
			\path (v1) edge (6.5,-1);
			\path (u) edge (7,-1);
			\path (u) edge (8,-1);
        \end{tikzpicture}

    \subcaption{A vertex from $O_n$  in the second neighborhood of $x\in V_d'$ gives still an advantage.}
    \label{fig:branching:dom-many-udo-nbs-with-one-On-nb}
\end{subfigure}
\quad 
\begin{subfigure}[b]{.55\textwidth}
\centering
	\begin{tikzpicture}[transform shape]

			\tikzset{every node/.style={ fill = black,circle,minimum size=0.3cm}}
			\node[draw,fill=gray,label={above:$x$}] (x) at (7,0) {};
			\node[draw,fill=gray,label={above:$y$}] (v) at (9,0) {};
			\node[draw,rectangle,fill=black,label={above:$S_x$}] (s1) at (5.5,0) {};
			\node[draw,rectangle,fill=black,label={above:$S_y$}] (s2) at (10.5,0) {};
			\node[fill=none,label={below:$\cdots$}] at (7,-0.3) {};
			\node[fill=none,label={below:$\cdots$}] at (9,-0.3) {};
			\path (x) edge[-] (s1);
			\path (x) edge (v);
			\path (v) edge[-] (s2);
			\path (x) edge (6.5,-1);
			\path (x) edge (7.5,-1);
			\path (v) edge (8.5,-1);
			\path (v) edge (9.5,-1);
        \end{tikzpicture}

    \subcaption{$N_S(x)$ belongs to the same connected component $S_x$ of $G[S]$, and so does $N_S(y)$ belong to $S_y$, but $S_x\neq S_y$.}
    \label{fig:branching:dom-nbs-diff-compos}
\end{subfigure}
    
    \caption{Branching Rules~\ref{branching:dom-many-udo-nbs-with-one-On-nb} and~\ref{branching:dom-nbs-diff-compos}}
\end{figure}
\begin{brarule}
\label{branching:dom-nbs-diff-compos}
Let $x,y\in V'_d$, $xy\in E$, such that $z\in\{x,y\}$ is adjacent to a connected component $S_z$ of $G[S]$, with $S_x\neq S_y$, 
as illustrated in \autoref{fig:branching:dom-nbs-diff-compos}. 
Then branch as follows.
\begin{enumerate}
    \item Put $x$ in $O_d$.
    \item Put $x$ in $S$ and $y$ in $O_d$.
    \item Put $x$ in $S$ and $y$ in $S$.
\end{enumerate}
\end{brarule}

\begin{lemma}
The branching of rule~\ref{branching:dom-nbs-diff-compos} is a complete case distinction. Moreover, it leads to a branching vector that is not worse than
$(1,2,2+\delta).$
\end{lemma}

\begin{proof}
    The case distinction is clearly complete. In the third branch that corresponds to the last case in the case distinction, the number of connected components of $G[S]$ decreases by~$1$, so that  the measure decreases by $2+\delta$ 
    in the last component of the branching vector.  
\qed \end{proof}

\begin{redrule}\label{reduction:dom-2nonadjac-nbs-one-dom}
If $x,y\in V'_d$
and $xy\in E$, then delete the edge $xy$; see \autoref{fig:reduction:dom-2nonadjac-nbs-one-dom}.
\end{redrule}

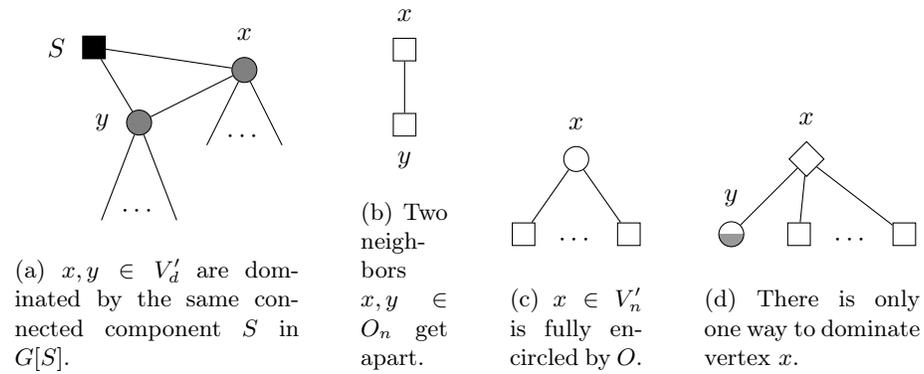
\begin{figure}[bt]
    \centering
\begin{subfigure}[b]{.31\textwidth}
    \centering
	\begin{tikzpicture}[transform shape,fill lower half/.style={path picture={\fill[#1] (path picture bounding box.south west) rectangle (path picture bounding box.east);}}]
			\tikzset{every node/.style={ fill = black,circle,minimum size=0.3cm}}
			\node[draw,fill=gray,label={above:$x$}] (u) at (7,0) {};
			\node[draw,fill=gray,label={left:$y$}] (v1) at (5.6,-.7) {};
			\node[draw,rectangle,fill=black,label={left:$S$}] (s) at 
			(5,.3) {};
			\node[fill=none,label={below:$\cdots$}] at (5.6,-1.3) {};
			\node[fill=none,label={below:$\cdots$}] at (7,-0.3) {};

			\path (u) edge[-] (v1);
			\path (u) edge[-] (s);
			\path (v1) edge[-] (s);
			\path (v1) edge (5.1,-2) edge (6.1,-2);
			\path (u) edge (6.5,-1) edge (7.5,-1);
			
        \end{tikzpicture}

    \subcaption{$x,y\in V_d'$ are dominated by the same connected component~$S$ in $G[S]$.}
    \label{fig:reduction:dom-2nonadjac-nbs-one-dom}
    \end{subfigure}
\qquad     
\begin{subfigure}[b]{.095\textwidth}
    \centering
    	
	\begin{tikzpicture}[transform shape]
			\tikzset{every node/.style={draw, fill = black,circle,minimum size=0.3cm}}
			\node[rectangle,fill=white,label={above:$x$}] (u) at (7,0) {};
			\node[rectangle,fill=white,label={below:$y$}] (v1) at (7,-1) {};
			\path (u) edge[-] (v1);
		
        \end{tikzpicture}

    \subcaption{Two neighbors $x,y\in O_n$ get apart.}
    \label{fig:reduction:edge-On}
\end{subfigure}
\qquad 
\begin{subfigure}[b]{.145\textwidth}
    \centering
    	
	\begin{tikzpicture}[transform shape]
			\tikzset{every node/.style={ fill = black,circle,minimum size=0.3cm}}
			\node[draw,fill=white,label={above:$x$}] (u) at (7,0) {};
			\node[draw,rectangle,fill=white,label={}] (v1) at (6.3,-1) {};
			\node[draw,rectangle,fill=white,label={}] (v2) at (7.7,-1) {};
			\node[fill=none,label={below:$\cdots$}] at (7,-0.5) {};
			\path (x) edge (v1);
			\path (x) edge (v2);
        \end{tikzpicture}

    \subcaption{$x\in V_n'$ is fully encircled by~$O$.}
    \label{fig:reduction:On-encircled}
\end{subfigure}
\qquad 
\begin{subfigure}[b]{.235\textwidth}
    \centering
    	
	\begin{tikzpicture}[transform shape]
			\tikzset{every node/.style={ fill = black,circle,minimum size=0.3cm},fill lower half/.style={path picture={\fill[#1] (path picture bounding box.south west) rectangle (path picture bounding box.east);}}}
			\node[draw,diamond,fill=white,label={above:$x$}] (u) at (6.4,0) {};
			\node[draw,rectangle,fill=white,label={}] (v1) at (6.3,-1) {};
			\node[draw,fill=white,fill lower half=black!40,label={$y$}] (y) at (5.4,-1) {};
			\node[draw,rectangle,fill=white,label={}] (v2) at (7.7,-1) {};
			\node[fill=none,label={below:$\cdots$}] at (7,-0.5) {};
			\path (u) edge (v1);
			\path (u) edge (y);
			\path (u) edge (v2);
        \end{tikzpicture}

    \subcaption{There is only one way to dominate vertex~$x$.}
    \label{fig:reduction:udo-only-one-dom}
\end{subfigure}

    \caption{Illustrating Reduction Rules~\ref{reduction:dom-2nonadjac-nbs-one-dom},~\ref{reduction:edge-On},~\ref{reduction:On-encircled} and~\ref{reduction:udo-only-one-dom}.}
\end{figure}

\begin{lemma}
Reduction Rule~\ref{reduction:dom-2nonadjac-nbs-one-dom} is sound and the measure does not change.
\end{lemma}

\begin{proof}
Let $M$ be a minimal CDS of $G$ such that $M \setminus V' =S$. Define $e =xy$ and  $\widetilde{G} = (V,E\setminus\lbrace e \rbrace)$. Now we want to show that~$M$ is also a minimal CDS in~$\widetilde{G}$. Since $x$ and~$y$ are already dominated by~$S$, the deletion of the edge~$e$ would not affect domination, nor could $x$ ever be the private neighbor of $y$ or vice versa. 
The connectivity is only important if $x,y\in M$. Vertices $x,y$ are dominated by the same connected component of~$S$, as otherwise Branching Rule~\ref{branching:dom-nbs-diff-compos} would have applied with priority. Hence, there exists a path $p = (x,p_1, \ldots, p_l, y)$, with internal vertices in~$S$.  Let $q=(q_1,\ldots,q_k)$ be a path in $G[M]$ such that there exists an $i\in \lbrace 1,\ldots, k-1\rbrace$ with $q_i=x$ and $q_{i+1}= y$. Then $\widetilde{q}= (q_1,\ldots, q_i, p_1,\ldots,p_l, q_{i+1},\ldots, q_{k})$ is a walk in $\widetilde{G}[M]$. 
Thus, $\widetilde{G}[M]$ is connected and $M$ is a CDS of $\widetilde{G}$. As $\widetilde{G}$ is a partial graph, 
$M$ is also a minimal CDS of $\widetilde{G}$ by \autoref{cor-minCDS-on-partialgraphs}.
\qed \end{proof}

\begin{redrule}\label{reduction:isolates}
If $x$ is an isolated vertex in $G[V'\cup O_n]$, do the following:
\begin{itemize}
    \item If $x$ is dominated, put~$x$ into $O_d$.
    \item If $x$ is not dominated, skip this branch. (This will always happen if $x\in O_n$.)
\end{itemize}
\end{redrule}

\begin{lemma}
Reduction Rule~\ref{reduction:isolates} is sound and never increases the measure.
\end{lemma}

\begin{proof}
As we are looking for a CDS, we cannot put an isolated vertex into a solution for reasons of domination (as 
a single vertex does not dominate the whole graph). As we have checked if $x\in V'_d$ should be put into the solution for  connectivity reasons by Branching Rule~\ref{branching:dom-diff-compos}, the only choice we have is to put~$x$ into $O_d$ if $x\in V'_d$. 
 \qed \end{proof}

\begin{redrule}\label{reduction:edge-On}
Let $x,y\in O_n$ be with $xy\in E$; see \autoref{fig:reduction:edge-On}.
We delete the edge $xy$.
\end{redrule}

\begin{lemma}
Reduction Rule~\ref{reduction:edge-On}  
is sound and does not change the measure.
\end{lemma}
\begin{proof}
    Let $M$ be a minimal CDS of $G$ such that $M \setminus V' =S$. Hence, $x,y\notin M$ and $M$ is a CDS of $\widetilde{G}:=G-xy = (V,E\setminus \{uv\})$. By \autoref{cor-minCDS-on-partialgraphs}, $M$ is also a minimal CDS of~$\widetilde{G}$.
\qed \end{proof}

\begin{redrule}\label{reduction:On-encircled}
If $x\in V'_n$ obeys 
$N(x)\cap V'=\emptyset$,
then skip this branch (\autoref{fig:reduction:On-encircled}).
\end{redrule}

\begin{lemma}
Reduction Rule~\ref{reduction:On-encircled} is sound.
\end{lemma}

\begin{proof}
If $x\in V'_n$ is put into the solution~$S$, then (as we excluded CDS with only one vertex) it needs a neighbor from $N_{V'\cup O_n}(x)$ within $S$ for connectivity reasons, as $x\notin N_G(S)$ and as we are looking for a CDS of size more than one. Similarly, vertex~$x$ needs a neighbor from $N_{V'\cup O_n}(x)$ within $S$ if $x$ is not put
into the solution. As we assume $N_{V'\cup O_n}(x)\subseteq O_n$, this is impossible, which is justifying to discard this branch.
\qed \end{proof}

\begin{redrule}\label{reduction:udo-only-one-dom}
Let $x\in V'_n\cup O_n$, with $N_{V'}(x)=\{y\}$.
Then, put $y$ into~$S$.
\end{redrule}

\begin{lemma}
Reduction Rule~\ref{reduction:udo-only-one-dom} is sound and the measure decreases by at least $1-\delta$.
\end{lemma}

\begin{proof}
    Let $M$ be a minimal CDS of $G$ such that $M \setminus V' =S$. Vertex~$x$ is dominated by~$M$. Either $x\notin M$ and then $y$ is the only vertex of $V'\cup S$ that can dominate~$x$, therefore $y\in M$, or $x\in M$ and since $G[M]$ is connected, $x$ has a neighbor in $M$, but since $y$ is the only vertex of  $V'\cup S$ in $N(x)$, we conclude $y\in M$.  If $y\in V'_d$, then the measure decreases by $1$, while if $y\in V'_n$, then the measure decreases by $1-\delta$.
\qed \end{proof}

\begin{brarule}\label{branching:udo-2-no-On}
Let $u\in V'_n$ with $\deg_{V'}(u)=2$, $\deg_{O_n}(u)=0$ and $N_{V'}(u)=\{v_1,v_2\}$; see \autoref{fig:branching:udo-2-no-On}. Then branch as follows.

\begin{enumerate}
    \item Put $u$ in $O_d$, $v_1$ in $S$.
    \item Put $u$ in $O_d$, $v_1$ in $O$, $v_2$ in $S$.
    \item Put $u$ in $S$, $v_1$ in $S$.
    \item Put $u$ in $S$, $v_1$ in $O_d$, $v_2$ in $S$.
\end{enumerate}
\end{brarule}

\noindent
More precisely, in the second branch, we put $v_1$ into $O_d$ if $v_1v_2\in E$ or if $v_1$ was already dominated and we put $v_1$ into $O_n$, otherwise. The same is done in the Branching Rules \ref{branching:On-2-no-On}, \ref{branching:dom-one-udo-nb-with-no-On-nbs}, \ref{branching:dom-2nbs-udo-common-nb}, \ref{branching:dom-2nbs-udo-nb-y} and \ref{branching:dom-2nbs-with-many-nbs}, as it can be decided whether the vertex goes into $O_n$ or into $O_d$.

\begin{figure}[tb]
    \centering
    	
\begin{subfigure}[b]{.3\textwidth}
    \centering
	\begin{tikzpicture}[transform shape,fill lower half/.style={path picture={\fill[#1] (path picture bounding box.south west) rectangle (path picture bounding box.east);}}]
			\tikzset{every node/.style={draw, fill = black,circle,minimum size=0.3cm}}
			\node[fill=white,label={above:$u$}] (u) at (7,0) {};
			\node[fill=white,fill lower half=black!40,label={below:$v_1$}] (v1) at (6,-1) {};
			\node[fill=white,fill lower half=black!40,label={below:$v_2$}] (v2) at (8,-1) {};
			\path (u) edge[-] (v1);
			\path (u) edge[-] (v2);
		\path(v2) edge[dashed] (v1);
        \end{tikzpicture}

    \subcaption{Branching if $u\in V_n'$ has two neighbors in $V'$ and none in $O_n$.}
    \label{fig:branching:udo-2-no-On}
    \end{subfigure}
        \quad 
\begin{subfigure}[b]{.3\textwidth} 
    \centering
    	
	\begin{tikzpicture}[transform shape,fill lower half/.style={path picture={\fill[#1] (path picture bounding box.south west) rectangle (path picture bounding box.east);}}]
			\tikzset{every node/.style={draw, fill = black,circle,minimum size=0.3cm}}
			\node[rectangle,fill=white,label={above:$u$}] (u) at (7,0) {};
			\node[fill=white,fill lower half=black!40,label={below:$v_1$}] (v1) at (6,-1) {};
			\node[fill=white,fill lower half=black!40,label={below:$v_2$}] (v2) at (8,-1) {};
			\path (u) edge[-] (v1);
			\path (u) edge[-] (v2);
\path(v2) edge[dashed] (v1);
        \end{tikzpicture}

    \subcaption{A similar rule for $u\in O_n$.}
    \label{fig:branching:On-2-no-On}
\end{subfigure}
\quad
\begin{subfigure}[b]{.31\textwidth}
    \centering
    	
	\begin{tikzpicture}[transform shape,fill lower half/.style={path picture={\fill[#1] (path picture bounding box.south west) rectangle (path picture bounding box.east);}}]
			\tikzset{every node/.style={ fill = black,circle,minimum size=0.3cm}}
			\node[draw,fill=gray,label={above:$u$}] (u) at (6,-1) {};
			\node[draw,fill=white,label={above:$v$}] (v1) at (7,-1) {};
			\node[fill=none,label={below:$\cdots$}] at (7,-1.3) {};

			\node[draw,fill=white,fill lower half=black!40] (y1) at (6.5,-2) {};
			\node[draw,fill=white,fill lower half=black!40] (y2) at (7.5,-2) {};
			\path (u) edge[-] (v1);
			\path (v1) edge[-] (y2);
			\path (v1) edge[-] (y1);
			
        \end{tikzpicture}

    \subcaption{We branch on $u$ and on~$v$ if $u\in S$. Note: $\deg_{V'\cup O_n}(v)\geq 3$ by Observation~\ref{obs:udo-deg-3}.}
    \label{fig:branching:dom-one-udo-nb-with-no-On-nbs}
\end{subfigure}

     \caption{Branching Rules~\ref{branching:udo-2-no-On},~\ref{branching:On-2-no-On} and~\ref{branching:dom-one-udo-nb-with-no-On-nbs}.}
\end{figure}
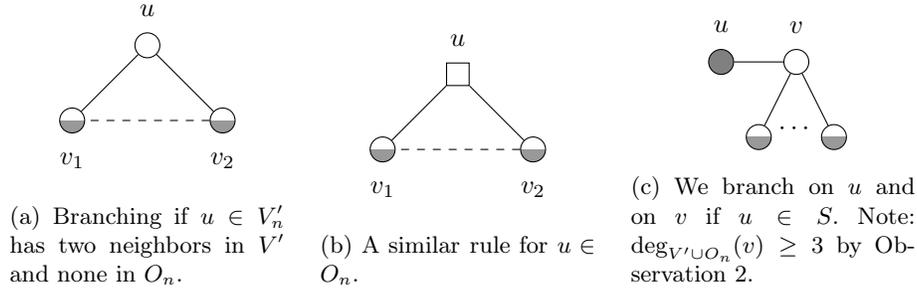

\begin{lemma}
The branching of rule~\ref{branching:udo-2-no-On} is a complete case distinction. Moreover, it leads to a branching vector that is not worse than
$(2-\delta,3-\delta-\alpha,2-\delta,3-\delta)\,.$
\end{lemma}

\begin{proof}
    Let $M$ be a minimal CDS of $G$ such that $M \setminus V' =S$. Then, either $u\notin M$ or $u\in M$. If $u\notin M$, as $u$ is dominated by $M$, and no vertex of $S$ dominates $u$, the vertex dominating $u$ is either $v_1$, or if it is not, it is~$v_2$.
    If $u\in M$, as $M$ is connected and not of size one, then vertex~$u$ is adjacent to another vertex of $M$. As no vertex of~$S$ is adjacent to $u$, either $v_1\in M$, or  $v_1\notin M$ and $v_2\in M$. 
     In all branches, in the worst case, a new connected component of~$S$ is created, which is why a $\delta$ is subtracted in each component of the branching vector. In the second branch, in the worst case, $v_1$ is not dominated by $S$ nor by~$v_2$, and thus moving~$v_1$ to~$O_n$ reduces the measure by $1-\alpha$. 
\qed \end{proof}

\begin{brarule}
\label{branching:On-2-no-On}
Let $u\in O_n$  with $N_{V'}(u)=\{v_1,v_2\}$ (\autoref{fig:branching:udo-2-no-On}). 
Then branch as follows.
\begin{enumerate}
    \item Put  $v_1$ in $S$, and thus $u$ in $O_d$.
    \item Put  $v_1$ in $O$, $v_2$  in $S$, and thus $u$ in $O_d$.
\end{enumerate}
\end{brarule}

\begin{lemma}
The branching of rule~\ref{branching:On-2-no-On} is a complete case distinction. Moreover, it leads to a branching vector that is not worse than
$(1+\alpha-\delta,2-\delta)\,.$
\end{lemma}

\begin{proof}

    Let $M$ be a minimal CDS of $G$ such that $M \setminus V' =S$. So $u\notin M$, $M$ dominates $u$, but $S$ does not, so either $v_1\in M$, or $v_1 \notin M$ and $v_2\in M$. In both cases $u\notin S$ and $u$ is dominated, thus $u\in O_d$.
     
    In the first branch, in the worst case, $v_1$ is not dominated by $S$, so it creates a new connected component for~$S$, $v_1$ goes from $V'$ to $S$ and $u$ from $O_n$ to $O_d$, so the measure decreases by $1+\alpha -\delta$. In the second branch, in the worst case, $v_1$ and $v_2$ are not dominated by $S$, so it creates a new connected component for $S$, $v_1$ goes from $V'$ to $O_n$, $v_2$ from $V'$ to~$S$ and $u$ from $O_n$ to~$O_d$, so the measure decreases by $1-\alpha+1+\alpha -\delta=2-\delta$. 
\qed \end{proof}

\begin{observation}\label{obs:udo-deg-3}
For each rule below, as it is applied in particular after Reduction Rule~\ref{branching:udo-2-no-On} and Reduction Rule~\ref{reduction:udo-only-one-dom}, we observe: for any vertex $v\in V'_n\cup O_n$, we have $\deg_{V'\cup O_n}(v) \geq 3$.
\end{observation}

\begin{brarule}\label{branching:dom-one-udo-nb-with-no-On-nbs}
Let $u\in V'_d$ with  $N_{V'_n}(u)=\{v\}$ and  $\deg_{O_n}(v)=0$; see \autoref{fig:branching:dom-one-udo-nb-with-no-On-nbs}. Then branch as follows.
\begin{enumerate}
    \item Put $u$ in $O_d$.
    \item Put $u$ in $S$ and $v$ in $O_d$ and thus all the vertices of $N_{V'}(v)\setminus\{u\}$ in $O$. 
    \item Put $u$ in $S$ and $v$ in $S$.
\end{enumerate}
\end{brarule}

\begin{lemma}
The branching of rule~\ref{branching:dom-one-udo-nb-with-no-On-nbs} is a complete case distinction. Moreover, it leads to a branching vector that is not worse than
$(1,4-2\alpha,2)\,.$
\end{lemma}

\begin{proof}
    Let $M$ be a minimal CDS  of $G$ such that $M \setminus V' =S$. Either $u$ is not in $M$ or $u$ is in~$M$ and $v$ is not in $M$ or $u$ and $v$ are in $M$. In the case where $u\in M$ and $v\notin M$,  in $G[M]$ vertex~$u$ is only adjacent to vertices of the same connected component of $G[S]$ due to Reduction Rule~\ref{branching:dom-diff-compos}; so $u$ must have a private neighbor, this private neighbor has to be~$v$, as  $N_{V'_n}(u)=\{v\}$ and $N_{O_n}(u)=\emptyset$. So $N(v)\cap M = \{u\}$, thus none of the vertices of $N_{V'}(v)$ is in~$M$, as $M$ is a minimal CDS. 
    Thus in the second branch, in the worst case, the vertices of $N_{V'}(v)$, a set which is of size at least $2$ due to Observation~\ref{obs:udo-deg-3},
    are not dominated by $S$ and the measure then reduces by $1-\alpha$ for each vertex in~$N_{V'}(v)$. 
\qed \end{proof}

\begin{brarule}\label{branching:dom-2nbs-udo-common-nb}
Let $u\in V'_d$, with $N_{V'}(u)=\{v_1,v_2\}$ and $\deg_{O_n}(u)=0$ (\autoref{fig:branching:dom-2nbs-udo-common-nb}). If $N_{V'}(v_1)\cap N_{V'}(v_2)$ contains a vertex $y\in V'$ different from~$u$, then branch as follows.

\begin{enumerate}
    \item Put $u$ in $O_d$.
    \item Put $u$ in $S$, $v_1$ in $S$.
    \item Put $u$ in $S$, $v_1$ in $O_d$, $v_2$ in $S$.
    \item Put $u$ in $S$, $v_1$ in $O_d$, $v_2$ in $O_d$ and thus $y$ in $O$. 
\end{enumerate}
\end{brarule}

\begin{lemma}
The branching of rule~\ref{branching:dom-2nbs-udo-common-nb} is a complete case distinction. Moreover, it leads to a branching vector that is not worse than
$(1,2,3,4-\alpha)\,.$
\end{lemma}

\begin{proof}
Let $M$ be a minimal CDS of $G$ such that $M \setminus V' =S$.  Either $u\notin M$ or $u\in M$. If $u\in M$, then either $v_1\in M$ or $v_1\notin M$, and in the latter case, either $v_2\in M$ or $v_2\notin M$.
Hence, the case distinction is complete. Assume now that $u\in M$ and $v_1,v_2\notin M$. As vertex~$u$ is adjacent to only one connected component of $G[S]$  due to Reduction Rule~\ref{branching:dom-diff-compos}, and as $N_M(u)=N_S(u)$, $u$ has a private neighbor. This private neighbor is either $v_1$ or $v_2$ and thus either $N(v_1)\cap M = \{u\}$ or $N(v_2)\cap M = \{u\}$, in both cases, since $y\in N(v_2)\cap N(v_1)$, $y\notin M$.
In the worst case, $y$ is not dominated and the measure in the last branch decreases by a total amount of $4-\alpha$.
\qed \end{proof}

\begin{figure}[tb]
    \centering
    	\begin{subfigure}[b]{.28\textwidth}
    \centering
	\begin{tikzpicture}[transform shape,fill lower half/.style={path picture={\fill[#1] (path picture bounding box.south west) rectangle (path picture bounding box.east);}}]
			\tikzset{every node/.style={ fill = black,circle,minimum size=0.3cm}}
			\node[draw,fill=gray,label={above:$u$}] (u) at (7,0) {};
			\node[draw,fill=white,label={left:$v_1$}] (v1) at (6,-1) {};
			\node[draw,fill=white,label={right:$v_2$}] (v2) at (8,-1) {};
			\node[draw,fill=white,fill lower half=black!40,label={above:$y$}] (y) at (7,-2) {};
			\node[fill=none,label={below:$\cdots$}] at (6,-1.3) {}; 
			\node[fill=none,label={below:$\cdots$}] at (8,-1.3) {};
			\path (u) edge[-] (v1);
			\path (u) edge[-] (v2);
			\path (v1) edge[-] (y);
			\path (v2) edge (y);
			\path (v1) edge (5.5,-2);
			\path (v1) edge (6.5,-2);
			\path (v2) edge (7.5,-2);
			\path (v2) edge (8.5,-2);
			\path (v2) edge[dashed] (v1);
        \end{tikzpicture}
        \subcaption{Possibly adjacent $v_1,v_2$ have $\geq 2$ common neighbors.}    \label{fig:branching:dom-2nbs-udo-common-nb}
\end{subfigure}
\quad 
\begin{subfigure}[b]{.26\textwidth}
    \centering
	\begin{tikzpicture}[transform shape,fill lower half/.style={path picture={\fill[#1] (path picture bounding box.south west) rectangle (path picture bounding box.east);}}]
			\tikzset{every node/.style={ fill = black,circle,minimum size=0.3cm}}
			\node[draw,fill=gray,label={above:$u$}] (u) at (7,0) {};
			\node[draw,fill=white,label={left:$v_1$}] (v1) at (6,-1) {};
			\node[draw,fill=white,label={right:$v_2$}] (v2) at (8,-1) {};
			\node[draw,fill=white,label={left:$y$},fill lower half=black!40] (y) at (6,-2) {};
			\node[fill=none,label={below:$\cdots$}] at (8,-1.3) {}; ;
			\path (u) edge[-] (v1);
			\path (u) edge[-] (v2);
			\path (v1) edge[-] (y);
			\path(v2) edge (v1);
			\path (v2) edge (7.5,-2);
			\path (v2) edge (8.5,-2);
        \end{tikzpicture}
    \subcaption{$v_1$ of adjacent  $v_1,v_2$ has only one more $V'$-neighbor:~$y$.}
    \label{fig:branching:dom-2nbs-udo-nb-y}
\end{subfigure}
\quad
\begin{subfigure}[b]{.36\textwidth}
    \centering
    	
	\begin{tikzpicture}[transform shape,fill lower half/.style={path picture={\fill[#1] (path picture bounding box.south west) rectangle (path picture bounding box.east);}}]
			\tikzset{every node/.style={ fill = black,circle,minimum size=0.3cm}}
			\node[draw,fill=gray,label={above:$u$}] (u) at (7,0) {};
			\node[draw,fill=white,label={left:$v_1$}] (v1) at (5.7,-1) {};
			\node[draw,fill=white,label={right:$v_2$}] (v2) at (8.3,-1) {};
			\node[draw,fill=white,fill lower half=black!40] (y1) at (4.9,-2) {};
			\node[draw,fill=white,fill lower half=black!40] (y2) at (6.6,-2) {};
			\node[draw,fill=white,fill lower half=black!40] (y3) at (7.4,-2) {};
			\node[draw,fill=white,fill lower half=black!40] (y4) at (9.1,-2) {};
			\node[fill=none,label={below:$\cdots$}] at (8.3,-1.3) {}; 
			\node[fill=none,label={below:$\cdots$}] at (5.7,-1.3) {};
			\path (u) edge[-] (v1);
			\path (u) edge[-] (v2);
			\path (v1) edge[-] (y1);
			\path (v1) edge[-] (y2);
			\path (v2) edge[-] (y3);
			\path (v2) edge[-] (y4);
			\path(v2) edge[dashed] (v1);
			
			\path (v1) edge (5.2,-2);
			\path (v1) edge (6.2,-2);
			\path (v2) edge (7.8,-2);
			\path (v2) edge (8.8,-2);
        \end{tikzpicture}

    \subcaption{When $v_1,v_2$ have $\geq 2$ `outside neighbors' each for better branching.}
    \label{fig:branching:dom-2nbs-with-many-nbs}
\end{subfigure}

    \caption{Branching Rules~\ref{branching:dom-2nbs-udo-common-nb},~\ref{branching:dom-2nbs-udo-nb-y} and~\ref{branching:dom-2nbs-with-many-nbs}}

\end{figure}
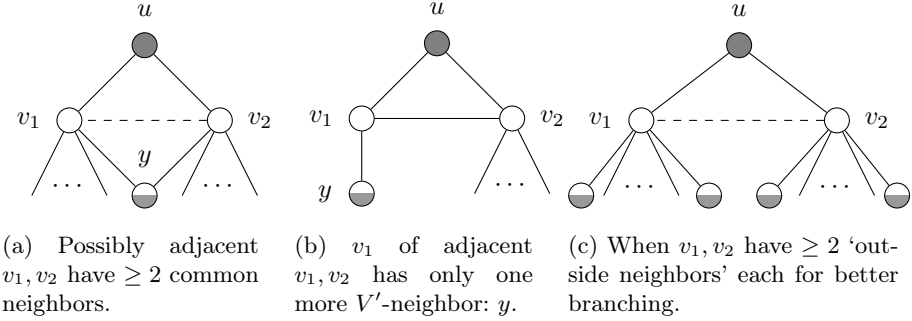

\begin{brarule}\label{branching:dom-2nbs-udo-nb-y}
Let $u\in V'_d$, $N_{V'_n}(u)=\{v_1,v_2\}$, 
$N_{V'_d}(u)=N_{O_n}(u)=\emptyset$ and $N_{V'}(v_1)\setminus \{u,v_2\}=\{y\}$; see \autoref{fig:branching:dom-2nbs-udo-nb-y}.
Then  branch as follows.

\begin{enumerate}
    \item Put $u$ in $O_d$.
    \item Put $u$ in $S$, $v_1$ in~$S$.
    \item Put $u$ in $S$, $v_1$ in $O_d$, $v_2$ in~$S$.
    \item Put $u$ in $S$, $v_1$ in $O_d$, $v_2$ in $O_d$ and $y$ in $O$.
    \item Put $u$ in $S$, $v_1$ in $O_d$, $v_2$ in $O_d$ and $y$ in $S$ and thus the vertices of $N_{V'}(v_2)\setminus \{u,v_1\}$ in~$O$. 
\end{enumerate}
\end{brarule}

\begin{lemma}
The branching of rule~\ref{branching:dom-2nbs-udo-nb-y} is a complete case distinction. Moreover, it leads to a branching vector that is not worse than
$(1,2,3,4-\alpha,5-\delta-\alpha)\,.$
\end{lemma}

\begin{proof} 
Note that $\{v_1,v_2\}\in E$, as drawn in \autoref{fig:branching:dom-2nbs-udo-nb-y}, since $\deg_{V'\cup O_n}(v_1)\geq 3$ by Observation~\ref{obs:udo-deg-3}. 
Let $M$ be a minimal CDS of~$G$ with $M \setminus V' =S$.  Either $u\notin M$ or $u\in M$. In the second case, we make further case distinctions: either $v_1\in M$ or $v_1\notin M$, and in the latter case, either $v_2\in M$ or $v_2\notin M$. Assume $u\in M$ and $v_1,v_2\notin M$. As $u$ is adjacent to only one connected component of $G[S]$ and $N_M(u)=N_S(u)$, the private neighbor of~$u$ is either $v_1$ or $v_2$ and thus either $N(v_1)\cap M = \{u\}$ or $N(v_2)\cap M = \{u\}$. This yields the last two branches. Consider $N(v_1)\cap M = \{u\}$: as  we have $y\in N(v_1)$,   $y\notin M$ follows.  If $N(v_1)\cap M \neq \{u\}$, we have  $N(v_2)\cap M = \{u\}$, so any vertex in $N(v_2)\setminus \{u\}$ is not in~$M$. But also, as $y$ is the only possible vertex of $N(v_1)$, apart from $u$, that can be in $M$, so it has to be in~$M$. 

Now, after Reduction Rule~\ref{branching:dom-2nbs-udo-common-nb}, we know that $y\notin N(v_2)$. Moreover, 
by Observation~\ref{obs:udo-deg-3},
$N(v_2)\setminus \{u,v_1\}\neq\emptyset$. 
As, in the worst case, $y$ is not dominated and none of the vertices of $N_{V'}(v_2)\setminus \{u,v_1\}$, the measure in the last two branches decreases by a total of $4-\alpha$ or by $ 3+|N_{V'}(v_2)\setminus \{u,v_1\}|\times(1-\alpha)+1-\delta \geq 3+1-\alpha+1-\delta = 5-\alpha-\delta$,  respectively.
 \qed \end{proof}

\begin{brarule}\label{branching:dom-2nbs-with-many-nbs}
(\autoref{fig:branching:dom-2nbs-with-many-nbs})
Let $u\in V'_d$, $N_{V'_n}(u)=\{v_1,v_2\}$, $N_{V'_d}(u)=N_{O_n}(u)=\emptyset$,  such that  $\left| N_{V'}(v_1)\setminus \{u,v_2\} \right|\geq 2$, as well as $\left| N_{V'}(v_2)\setminus \{u,v_1\} \right|\geq 2$. Then branch as follows.

\begin{enumerate}
    \item Put $u$ in $O_d$.
    \item Put $u$ in $S$, $v_1$ in $S$.
    \item Put $u$ in $S$, $v_1$ in $O_d$, $v_2$ in $S$.
    \item Put $u$ in $S$, $v_1$ in $O_d$, $v_2$ in $O_d$ and the vertices of $N_{V'}(v_1)\setminus \{u,v_2\}$ in $O$. 
    \item Put $u$ in $S$, $v_1$ in $O_d$, $v_2$ in $O_d$ and the vertices of $N_{V'}(v_2)\setminus \{u,v_1\}$ in $O$. 
\end{enumerate}
\end{brarule}

\begin{lemma}
The branching of rule~\ref{branching:dom-2nbs-with-many-nbs} is a complete case distinction. Moreover, it leads to a branching vector that is not worse than
$(1,2,3,5-2\alpha,5-2\alpha)\,.$
\end{lemma}

\begin{proof}
Let $M$ be a minimal CDS of~$G$ such that $M \setminus V' =S$.  Either $u\notin M$ (decreasing the measure by 1) or $u\in M$, and then either $v_1\in M$ (decreasing the measure by 2) or $v_1\notin M$. In this last case, either $v_2\in M$ (decreasing the measure by~3) or $v_2\notin M$. Assume $u\in M$ and $v_1,v_2\notin M$. As $u$ is adjacent to only one connected component of $G[S]$ and as $N_M(u)=N_S(u)$, $u$ has a private neighbor which is either $v_1$ or $v_2$ and thus either $N(v_1)\cap M = \{u\}$ or $N(v_2)\cap M = \{u\}$. This yields the last two branches: assume that $N(v_1)\cap M = \{u\}$. This implies, for each $y\in N(v_1)\setminus\lbrace u\rbrace$, that  $y\notin M$ and thus  $y\in O$ (decreasing the measure by at least $5-2\alpha$).  Assume that $N(v_1)\cap M \neq \{u\}$. This implies, for each $y\in N(v_2)\setminus\lbrace u\rbrace$,
that  $y\notin M$,  \emph{i.e.}, $y\in O$ (decreasing the measure by at least $5-2\alpha$). 
\qed \end{proof}

The next and final rule will never be applied when this algorithm is applied to a $2-$degenerate graph. It rather prepares the ground for the general case. However, with our current measure, this would yield an $\mathcal{O}(2^n)$ algorithm in the general case. With a modified measure, as described in the next section, we will achieve a better running time.

\begin{brarule}\label{branching:dom-general}
Let $x\in V'_d$. Then we branch as follows.
 
\begin{enumerate}
    \item Put $x$ in $O_d$.
    \item Put $x$ in $S$ and thus the vertices of $N_{V'}(x)$ in $V'_d$.
\end{enumerate}
\end{brarule}

\begin{lemma}
The branching of rule~\ref{branching:dom-general} is a complete case distinction. Moreover, it leads to a branching vector that is not worse than
$(1,1)\,.$
\end{lemma}

\begin{proof}
The case distinction is obviously complete. As it is applied after every previous rule, we have that $\deg(x)\geq 3$ and every vertex of $N(x)$ is in $V'_n$. In the second branch, this means at least $2$ vertices are going from $V'_n$ to $V'_d$.
\qed \end{proof}

\noindent We can summarize the lemmas that we proved so far with the following two statements.

\begin{lemma}\label{lem:branchingrules}
Each of the mentioned branching rules covers all cases in the described situation. The branching will lead to a branching vector as listed in \autoref{tab:branching-vectors-2degen}.
\end{lemma}

\begin{lemma}\label{lem:reductionrules}
The  mentioned reduction rules are sound and their application never increases the measure.
\end{lemma}
A case analysis shows the correctness of our algorithm in the following sense: 
\begin{lemma}\label{lem:covers-all}
The Reduction and Branching Rules cover all possible cases.
\end{lemma}

\begin{proof} 
As explained below, Reduction Rule~\ref{branching:dom-general} serves as a final catch-all. For the  2-degenerate case, we should focus on all other rules and prove that they cover all cases.
This means that we have to show that the proposed algorithm will resolve each 2-degenerate graph completely. Our rule priorities might remove vertices of arbitrary degree from the graph $G'$ that is a partial graph of the graph $G[V'\cup O_n]$. This way, we again arrive at a 2-degenerate graph. Yet, what is important for dealing with 2-degenerate graphs is to consider all cases of a vertex~$u$ of degree at most~2.  The degree conditions in the following case distinction refer to~$G'$. Degree-0 vertices are treated with Reduction Rule~\ref{reduction:isolates}. We now discuss vertices~$u$ of degrees one or two.
    There are two different cases: either $u$ is in $V'_n\cup O_n$, or $u\in V'_d$.
    \\[1ex]
    \underline{Case 1}: $u$ is not dominated by $S$. Now we  discuss its degree in $G'$, either it is $1$ or it is $2$.
    \\
     \underline{Case 1.1}: $\deg(u) = 1$ and we denote by $v$ the neighbor in $G'$ of $u$. If $v\in O_n$, then either $u$ is in $O_n$ and then satisfies the conditions of Reduction Rule~\ref{reduction:edge-On}, or it is in $V_n$ and thus satisfies the conditions of Reduction Rule~\ref{reduction:On-encircled}. If $v\notin O_n$, $v$ satisfies the conditions of Reduction Rule~\ref{reduction:udo-only-one-dom}.
 \\
     \underline{Case 1.2}: $\deg(u) = 2$. Now we discuss the number of neighbors of $u$ in $O_n$. If both neighbors are in $O_n$, Reduction Rule~\ref{reduction:On-encircled} applies if $u\in V'_n$ and Reduction Rule~\ref{reduction:edge-On} applies otherwise. If only one neighbor of $u$ is in $O_n$, then this means that if $u\in O_n$ Reduction Rule~\ref{reduction:edge-On} applies, and otherwise Reduction Rule~\ref{reduction:udo-only-one-dom} applies.
        If none of the neighbors are in $O_n$, then either $u\in V'_n$ and then Reduction Rule~\ref{branching:udo-2-no-On} applies, or otherwise $u\in O_n$ and thus Reduction Rule~\ref{branching:On-2-no-On} applies.
     \\[1ex]
    \underline{Case 2}: $u$ is dominated by $S$.
    \\
        \underline{Case 2.1}: $\deg(u) = 1$. We denote by $v$ the neighbor in $G'$ of $u$. If $v$ is dominated by $S$ then Reduction Rule~\ref{reduction:dom-2nonadjac-nbs-one-dom} applies. If $v$ is not dominated by $S$, then either it has no neighbors in $O_n$ and then Reduction Rule~\ref{branching:dom-one-udo-nb-with-no-On-nbs} applies, or it has at least one neighbor in $O_n$ and Reduction Rule~\ref{branching:dom-many-udo-nbs-with-one-On-nb} applies.
         \\
        \underline{Case 2.2}: $\deg(u) = 2$. We denote $\{v_1,v_2\} = N_{V'}(u)$. If either $v_1$ or $v_2$ is dominated, then Reduction Rule~\ref{reduction:dom-2nonadjac-nbs-one-dom} applies. If both of them are not dominated by $S$, either they have at least one neighbor in $O_n$ and then Reduction Rule~\ref{branching:dom-many-udo-nbs-with-one-On-nb} applies or they do not have any neighbor in~$O_n$. So we know that $deg_{V'}(v_1) \geq 3$, and $deg_{V'}(v_2) \geq 3$, otherwise we could apply one of the rules of the previous case to $v_1$ or~$v_2$. Now, either $v_1$ and $v_2$ have a common 
        neighbor outside of $u$ and Reduction Rule~\ref{branching:dom-2nbs-udo-common-nb} applies, or they do not and then either at least one of them has exactly one neighbor that is not $u$ or $v_1$ (or $v_2$, respectively) and Reduction Rule~\ref{branching:dom-2nbs-udo-nb-y} applies, or they both have at least two neighbors outside of $u,v_1$ and $v_2$, so that Reduction Rule~\ref{branching:dom-2nbs-with-many-nbs} applies.

We finally have to prove the correctness of the algorithm for a general graph. If any vertex satisfies the conditions of any Branching  or Reduction Rule apart from Reduction Rule~\ref{branching:dom-general}, then we apply such a rule, otherwise no such rule applies, which means that the minimum degree of $G'[V'\cup O_n]$ is $3$. If at least one vertex is in $V_d$, then we can apply Reduction Rule~\ref{branching:dom-general}. Assume it is not the case, that means $V_d$ is empty, so every vertex of $V'\cup O_n$ is not dominated (and it is not empty, otherwise the algorithm would have ended). Since in the beginning, $S$ was not empty, $S$ is never empty. Moreover, at any point in the algorithm, $N(S)= V_d\cup O_d$. So in our case, we have $N(S) = O_d$. This means that $S$ is not a dominating set (there are vertices in $V_n\cup O_n$) and $S$ cannot have any neighbor added, so there is no CDS $M\subset V$ such that $M \setminus V' = S$. This branch has to be discarded.
\qed \end{proof}

\begin{table}[tb]
    \centering
    \begin{tabular}{l|l|l}
        Branching Rule \# & Branching vector & Branching number\\\hline
        \ref{branching:dom-nbs-On} & $(1,1+\alpha)$ & always better than \ref{branching:dom-many-udo-nbs-with-one-On-nb}
\\
\ref{branching:dom-diff-compos} & $(1,1+\delta)$& always better than \ref{branching:dom-nbs-diff-compos}\\\ref{branching:dom-many-udo-nbs-with-one-On-nb} & $(1,2,2+\alpha)$& $ < 1.9766$
\\
\ref{branching:dom-nbs-diff-compos} & $(1,2,2+\delta)$& $< 1.9766$
\\
        \ref{branching:udo-2-no-On} &  $(2-\delta,3-\delta-\alpha,2-\delta,3-\delta)$ & $<1.8269$\\
\ref{branching:On-2-no-On} & $(1+\alpha-\delta,2-\delta)$& $<1.6420$\\

\ref{branching:dom-one-udo-nb-with-no-On-nbs} & $(1,4-2\alpha,2)$& $< 1.7691$
\\
\ref{branching:dom-2nbs-udo-common-nb} & $(1,2,3,4-\alpha)$& $ < 1.9333$
\\\ref{branching:dom-2nbs-udo-nb-y} & $(1,2,3,4-\alpha,5-\delta-\alpha)$& $ < 1.9767$
\\\ref{branching:dom-2nbs-with-many-nbs} & $(1,2,3,5-2\alpha,5-2\alpha)$& $ < 1.9420$
\\\ref{branching:dom-general} & $(\beta,3-2\beta)$ where $\beta=1$ & $= 2$ \qquad \tiny not for 2-degenerate graphs 
\end{tabular}
    \caption{Collection of all branching vectors for the 2-degenerate case; the branching numbers are displayed for the different cases with $\alpha=\avalue$ and $\delta=\dvalue$.
    }
    \label{tab:branching-vectors-2degen}
\end{table}

\begin{theorem}
\textsc{Connected Dominating Set Enumeration} can be solved in time $\mathcal{O}(\cdsbasis^n)$ on 2-degenerate graphs of order~$n$, using polynomial space only.
The claimed branching number is attained by setting $\alpha=\avalue$ and $\delta=\dvalue$; see \autoref{tab:branching-vectors-2degen}.
\end{theorem}

\begin{proof}
For the correctness of the reduction and branching rules, we refer to the previous lemmas. For the running time, consider  \autoref{tab:branching-vectors-2degen} where we collect all branching vectors that occur in our analysis. In the table, one can also see the branching numbers for the different cases; sometimes, it is clear that other cases are always better. Namely, $(1,2,2+\gamma)$ can be viewed as a $(1,1)$-branch followed by one $(1,1+\gamma)$-branch in one of the two branches. As we do binary branchings only (put a vertex in the solution or out), it is clear that in our search tree, no solution can be enumerated twice. So, as the depth of our search tree is linear, only polynomial space is needed. There is caveat concerning Reduction Rule~\ref{branching:dom-2nbs-with-many-nbs}: the fourth and fifth branch is not binary in the strict sense. Therefore, our recursion would need additional information when calling this fifth branch, as one of the neighbors considered in the fourth branch not to be inside the solution must now be in the solution to prevent double enumeration.
\qed \end{proof}

\begin{corollary}
2-degenerate graphs have no more than $\mathcal{O}(\cdsbasis^n)$  minimal CDS.
\end{corollary}

\begin{proof}
As our algorithm for 2-degenerate graphs produces all minimal connected dominating sets of such graph, the upper-bound on its running-time implies an upper-bound on the number of outputted objects.
\qed \end{proof}

\begin{remark}
We could deduce a corresponding CDS enumeration result for subcubic graphs, as they can be dealt with as 2-degenerate graphs after the initial branching.
There is a clear indication that the bound that we derive in this way is not tight.
Namely, Kangas \emph{et al.} have shown in~\cite{KanKKK2018} that there are no more than $1.9351^n$ many connected sets of vertices in a subcubic graph of order~$n$.
\end{remark}

\section{Getting $\beta$ into the game: the general case}\label{sec-algobeta}

As indicated above, we are now refining the previous analysis by splitting the set of vertices of the currently considered graph further.
More precisely, the set of vertices $V'$ that have not been decided to come into or to be out of the solution~$S$ is split into the set of vertices~$V'_n$ that is still undominated and the set $V'_d$ of vertices that is already dominated. From the viewpoint of the original graph, the neighbors of the solution set $S$ are therefore partitioned into the sets $V'_d$ and $O_d$.
However, observe that although we do not consider $O_d$ anymore in the present graph, we do care about $V'_d$, since $V'_d$-vertices can still be either placed into $S$ or into $O_d$ by future branching or reduction rules. This is also reflected in the measure of the instance~$I$, which is now defined as:

$$\mu(I)=|V'_n|+\alpha\cdot |O_n|+\beta\cdot |V'_d|+\delta \cdot c$$

We set $\alpha=0.110901$, $\beta = 0.984405 $ and $\delta = 0.143516 $.
We are next working through the branching rules one by one, as in particular the branching vectors will now split off into different cases, as in our preliminary analysis, we only took care of the worst cases, not differentiating between 
$V'_n$ or $V'_d$ (which was summarized under the set name $V'$ before).
For the convenience of the reader, we repeat the formulation of the branching rules in the following, adapting the notations.

We start with a table stating the branching vectors according to this new measure for some branching rules for which it is straightforward. The last two branchings are worst cases.

\begin{center}

\begin{tabular}{l|l|l||l|l|l}
    Rule & Branching vector & Br. number
    &Rule & Branching vector & Br. number\\\hline
     \# \ref{branching:dom-nbs-On}&  $(\beta,\beta+\alpha)$ &  $< 1.9489$ &
     \# \ref{branching:dom-diff-compos} & $(\beta,\beta+\delta)$  & $<1.9297$ \\
     \# \ref{branching:dom-many-udo-nbs-with-one-On-nb} &$(\beta,\beta+1,\beta+1+\alpha) $ & $<1.9896$ &
     \# \ref{branching:dom-nbs-diff-compos} & $(\beta,2\beta,2\beta+\delta)$ & $< 1.9896$ \\

\end{tabular}
\end{center}

\noindent\underline{Reduction Rule~\ref{branching:udo-2-no-On}} Let $u\in V'_n$ with  $\deg_{V'}(u)=2$, $\deg_{O_n}(u)=0$ and $N_{V'}(u)=\{v_1,v_2\}$. The branching vector is different when $v_1,v_2\in V'_d$ or $v_1,v_2\in V'_n$. We will assume $v_1v_2\notin E$, as this is always the worst case. A similar analysis applies to 
\underline{Reduction Rule~\ref{branching:On-2-no-On}}.

\begin{center}\scalebox{.85}{\begin{tabular}{l||l|l||l|l}
Condition & Branching vector for \# \ref{branching:udo-2-no-On} & Br. 
number & Br. vector for \# \ref{branching:On-2-no-On} & Br. 
number  \\\hline
    $v_1,v_2\in V'_n$ &           $(2-\delta,3-\delta-\alpha, 2-\delta,3-\delta)$& $<1.8463$&           $(1+\alpha-\delta,2-\delta)$& $<1.6635$ \\
    $v_1\in V'_n$, $v_2\in V'_d$ & $(2-\delta,2-\alpha+\beta,2-\delta,2+\beta)$& $<1.8236$& $(1+\alpha,1+\beta)$& $<1.5855$\\
    $v_1,v_2\in V'_d$ &           $(1+\beta,1+2\beta,1+\beta,1+2\beta)$ & $<1.7785$&           $(\beta+\alpha,2\beta+\alpha)$ & $<1.5817$ \\
\end{tabular}}
\end{center}

\noindent\underline{Reduction Rule~\ref{branching:dom-one-udo-nb-with-no-On-nbs}} We look at $u\in V'_d$, $\{v\}= V'_n\cap N(u)$ and the neighbors of $v$; we know that $N(v)$ contains at least $u,v_1,v_2\notin O$; differentiating between $v_i\in V_n'$ or $v_i\in V_d'$ (as before), we arrive at three cases never leading to a branching number worse than $1.78$.

Details of case distinctions for Reduction Rule~\ref{branching:dom-one-udo-nb-with-no-On-nbs}:
\begin{center}
\begin{tabular}{l|l|l}
Condition & Branching vector & Branching number  \\\hline
    $v_1,v_2\in V'_n$ &           $(\beta,\beta+3-2\alpha,1+\beta)$& $< 1.7796$ \\
    $v_1\in V'_n$, $v_2\in V'_d$ & $(\beta,2\beta+2-\alpha,1+\beta)$& $<1.7729$\\
    $v_1,v_2\in V'_d$ &           $(\beta,3\beta+1,1+\beta)$ & $<1.7665$\\
\end{tabular}
\end{center}

\noindent\underline{Branching Rule~\ref{branching:dom-2nbs-udo-common-nb}} 
Let $u\in V'_d$ be with $\deg_{V'\cup O_n}(u)=2$, such that $v_1,v_2\in N_{V'_n}(u)$ and $N_{V'}(v_1)\cap N_{V'}(v_2)$ contains a vertex $y\in V'$ different from~$u$. We differentiate $y\in V_n'$ or $y\in V_d'$ and arrive at a branching number not worse than $1.9403$.

Details of the analysis of Branching Rule~\ref
{branching:dom-2nbs-udo-common-nb}:
\begin{center}
  \begin{tabular}{l|l|l}
        Condition & Branching vector & Branching number\\\hline
        $y\in V'_n$ & $(\beta,2,3,3+\beta-\alpha)$ & $<1.9403$ \\
        $y\in V'_d$ & $(\beta,1+\beta,2+\beta,2+2\beta)$& $<1.9398$\\
        \end{tabular}
\end{center}
Notice that in the case where $y\in V'_n$, in the second and third branches, $y$ is put in $V'_d$ and the measure associated to $y$ decreases by $(1-\beta)$.

\noindent\underline{Reduction Rule~\ref{branching:dom-2nbs-udo-nb-y}} Let $u\in V'_d$ be with $\deg_{V'\cup O_n}(u)=2$, such that $v_1,v_2\in N_{V'_n}(u)$ and  $ N_{V'}(v_1)\setminus \{u,v_2\}= \{y\}$. We distinguish $y\in V_n'$ and $y\in V_d'$; the first case leads to another worst-case branching for our algorithm.
In the last branch of each vector, $\min(\beta,1-\alpha)$ corresponds to whether $z\in N(v_2)\setminus\{u,v_1\}$ belongs to $V_d'$ or to $V_n'$. 

\begin{center}
  \begin{tabular}{l|l|l}
        Condition & Branching vector & Branching number\\\hline
        $y\in V'_n$ & $(\beta,1+\beta,2+\beta,3+\beta-\alpha,3+\beta-\delta+\min(\beta,1-\alpha))$& $<1.9896$\\
        $y\in V'_d$ & $(\beta,1+\beta,2+\beta,2+2\beta,2+2\beta+\min(\beta,1-\alpha))$ & $< 1.9813$ \\
        \end{tabular}
\end{center}

\noindent\underline{Reduction Rule~\ref{branching:dom-2nbs-with-many-nbs}}
Let $u\in V'_d$, $N_{V'_n}(u)=\{v_1,v_2\}$, $N_{V'_d}(u)=N_{O_n}(u)=\emptyset$. We further assume that $\left| N_{V'}(v_1)\setminus \{u,v_2\} \right|\geq 2$, as well as $\left| N_{V'}(v_2)\setminus \{u,v_1\} \right|\geq 2$.
The only thing to discuss here is whether the vertices of  $N_{V'}(v_1)\setminus \{u,v_2\} $ 
$ N_{V'}(v_2)\setminus \{u,v_1\}$ are in $V'_n$ or $V'_d$. This leads to nine subcases, which never produce a branching vector worse than $1.9453$.

We will denote by $n_1 =\left|N_{V'_n}(v_1)\setminus \{u,v_2\}\right|$ and $n_2 =\left|N_{V'_n}(v_2)\setminus \{u,v_1\}\right|$.
In the worst case, $\left|N_{V'}(v_1)\setminus \{u,v_2\}\right| =\left|N_{V'}(v_2)\setminus \{u,v_1\}\right| = 2$. Thus, we obtain the following branching vectors (where $n_1,n_2\in \{0,1,2\}$):

\begin{center}
  \begin{tabular}{ll|lllll|l}
        \multicolumn{2}{l}{Condition} &\multicolumn{4}{l}{Branching vector} & & Branching number\\\hline 
        $n_1=0$,& $n_2=0$ & $(\beta$ ,&$1+\beta,$&$2+\beta$,&$2+3\beta,$&$2+3\beta )$ & $<1.9430$\\
        $n_1 = 0$,& $n_2 =1$ & $(\beta,$&$1+\beta,$&$ 3,$&$ 2+3\beta,$&$ 3+2\beta -\alpha$)& $<1.9440$\\
        $n_1 = 0$, &$n_2 =2$ & $(\beta,$&$1+\beta,$&$ 4-\beta,$&$ 2+3\beta,$&$ 4+\beta -2\alpha$)&$<1.9453$\\
        
        $n_1 = 1$, &$n_2 =0$ & $(\beta,$&$2,$&$ 2+\beta,$&$ 3+2\beta-\alpha,$&$ 2+3\beta$)& $<1.9426$\\
        $n_1 = 1$, &$n_2 =1$ & $(\beta,$&$2,$&$ 3,$&$        3+2\beta -\alpha,$&$ 3+2\beta -\alpha$)&$<1.9437$\\
        $n_1 = 1$, &$n_2 =2$ & $(\beta,$&$2,$&$ 4-\beta,$&$  3+2\beta -\alpha,$&$ 4+\beta -2\alpha$)& $< 1.9449$\\
        
        $n_1 = 2$, &$n_2 =0$ & $(\beta,$&$3-\beta,$&$ 2+\beta,$&$ 4+\beta -2\alpha,$&$ 2+3\beta$)&$<1.9425$\\
        $n_1 = 2$, &$n_2 =1$ & $(\beta,$&$3-\beta,$&$ 3,$&$ 4+\beta -2\alpha,$&$ 3+2\beta-\alpha$)&$<1.9435$\\
        $n_1 = 2$, &$n_2 =2$ & $(\beta,$&$3-\beta,$&$ 4-\beta,$&$ 4+\beta -2\alpha,$&$ 4+\beta -2\alpha$)&$<1.9448$\\
    \end{tabular}

\end{center}

\noindent\underline{Reduction Rule~\ref{branching:dom-general}} When this branching rule applies, the sets in which the different vertices reside are all already known: as it is applied after
Reduction Rule~\ref{branching:dom-nbs-On} and Reduction Rule~\ref{branching:dom-nbs-diff-compos}, all the neighbors of $x\in V'_d$ are in $V'_n$, moreover $\deg_{V'_n}(x) = \deg_{V'\cup O_n}(x) \geq 3$ as it is applied after every other rule. The branching vector  $(\beta, 3-2\beta )$ gives a branching number $<1.9896$, which is the last worst-case branching.

\begin{theorem}
\textsc{Connected Dominating Set Enumeration} can be solved in time $\mathcal{O}(\genbasis^n)$ on graphs of order~$n$, using polynomial space only.
\end{theorem}

\begin{corollary}
There are no more than $\mathcal{O}(\genbasis^n)$ many minimal connected dominating sets in a graph of order $n$.
\end{corollary}

\section{Achieving polynomial delay is not easy}\label{sec-PolyDelay}
Unfortunately, we cannot guarantee that the time elapsed between producing two subsequent minimal solutions is polynomial only.
One reason for this is the fact that the corresponding extension problem is \NP-complete, as we can show by a reduction from \textsc{3-Sat}.
Namely, if \textsc{Connected Dominating Set Extension} would be solvable in polynomial time, then we might cut search tree branches whenever it is clear that no solution is to be expected beyond a certain node of the search tree, as we can associate to such a node also a set of vertices~$U$ that is decided to go into the solution. More precisely speaking, we know a bit more in each node of the search tree, so that we are actually looking at a more general question. But if \textsc{Connected Dominating Set Extension} is \NP-hard, then we cannot expect a polynomial-time algorithm for this more general question, either.
On the positive side, it has been recently exemplified with the enumeration problem of minimal Roman domination functions in~\cite{AbuFerMan2022} that a polynomial-time algorithm for the corresponding extension problem can be adapted so that polynomial delay can be achieved for this type of enumeration problem.

\begin{theorem}
The \textsc{Connected Dominating Set Extension} problem is \NP-complete, even when restricted to 2-degenerate graphs.
\end{theorem}

\begin{proof}
By ``guess-and-check,'' we know that $\textsc{Connected Dominating Set}$\linebreak[4]   $\textsc{Extension}\in\NP$. For the \NP-hardness, we will use \textsc{3-Sat}. Let $X=\{ x_1,\ldots,x_n\}$ be the set of the variables and $C= \{ C_1,\ldots, C_m\}$ be the set of clauses of an instance of \textsc{3-Sat}. This implies that for each $C_j\in C$, there exist $l_{j,1},l_{j,2},l_{j,3}\in \{ x_1, \neg x_1, \ldots, x_n,\neg x_n\}$ with $C_j= \{l_{j,1},l_{j,2},l_{j,3}\}$. 

Define the graph $G=(V,E)$ with
\begin{equation*}
    \begin{split}
        V = &\ \{ s\} \cup \{w_i, v_i, \overline{v}_i, y_i, \overline{y}_i \mid i\in \lbrace 1,\ldots,n\rbrace\}\cup \{a_j,b_j,c_j\mid j\in \lbrace 1,\ldots,m\rbrace\},\\
        E = & \left(\bigcup_{i=1}^n \{ \{v_i,s\},\{\overline{v}_i,s\}, \{ v_i, \overline{v}_i\}, \{v_i,y_i\},\{\overline{v}_i,\overline{y}_i\}, \{w_i,y_i\},\{w_i,\overline{y}_i\} \}\right) \cup \\
        & \Biggl( \bigcup_{i=1}^n \{ \{b_j,a_j\},\{a_j,c_j\} \} \cup \{\{ b_j, v_i\}\mid  x_i \in \{l_{j,1},l_{j,2}\}\} 
        \\
        & \cup \{\{ b_j, \overline{v}_i\}\mid \neg x_i \in \{l_{j,1},l_{j,2}\}\} \cup \left\{\{ a_j, v_i\} \mid  x_i = l_{j,3}\right\} \\
        & \cup \{\{ a_j, \overline{v}_i\} \mid \neg x_i = l_{j,3}\}\Biggr)
    \end{split}
\end{equation*}
and the sets $U:= \{w_1,\ldots,w_n\}$, $V':=\{v_1, \overline{v}_1, \ldots, v_n, \overline{v}_n\}$. Clearly, this graph can be produced in polynomial time, since there are $5n + 3m + 1$ vertices and $7n + 5m$ edges. 

Before proving the correction of the reduction itself, let us first focus on the graph property claimed about the produced graph.

\begin{claim}
$G$ is 2-degenerate.
\end{claim}

To prove this claim, we will show that there exists a sequence of $6n + 3m + 1$ many vertices  $( v_1,\ldots,v_{6n + 3m + 1})\in V^{6n + 3m + 1}$ such that $v_k$ has at most degree~2 on the induced graphs $G[\{v_k,\ldots,v_{6n+3m+1}\}]$ for each $k$ from $6n+3m+1$ down to~$1$. 

Since $\deg(c_j) = 1$, $\deg(w_i) = \deg(y_i) = \deg(\overline{y}_i) = 2$ holds for each $i \in \{ 1, \ldots, n\}$ and $j\in \{1,\ldots, m\}$, these vertices always fulfil the condition of the sequence. Therefore, we define $M_1:= V\setminus \{c_j\mid j\in \{1,\ldots,m\} \} \cup \{ w_i, y_i,\overline{y}_i\mid i\in \{1,\ldots,n\}\}$.  For each $j\in \{1,\ldots, m\}$, $N(a_j) \subseteq \{ b_j,c_j\}\cup V'$ holds with $\vert N(a_j)\cap V'\vert =1$. Therefore, $a_j$ has degree~2 on $G[M_1]$. For each $j \in \{1,\ldots,m\}$, $\vert N(b_j)\setminus \{a_j\} \vert = 2$ holds. Define $M_2:= M_1\setminus \{a_1,b_1,\ldots,a_m,b_m\}$. This implies $N[v_i]\setminus M_2=N[\overline{v}_i]\setminus M_2= \{v_i,\overline{v}_i,s\}$
for each $i\in \{1,\ldots,n\}$. This leaves us with~$s$ finally. Thus, there exists a sequence that satisfies the claim and $G$ is 2-degenerate.  \hfill$\Diamond$

Now, we turn towards the correctness of the reduction. The main ingredients in terms of gadgets are displayed in \autoref{fig:ext-cds:basic-block}.

\begin{figure}[tb]
\centering
\begin{subfigure}[l]{.5\textwidth}\centering
    \begin{tikzpicture}[->,
scale=1, every node/.style={anchor=center, scale=0.7},node distance=1cm,main node/.style={circle,fill=white!20,draw,font=\sffamily\Large\bfseries}]
\node [main node] (c_j) at (0,0) {$c_j$};
\node [main node] (b_j) at (2,0.5) {$b_j$};
\node [main node] (a_j) at (1,0) {$a_j$};
\node  [main node] (l_j1) at (3,1) {$l_{j,1}$};
\node  [main node] (l_j2) at (3,0) {$l_{j,2}$};
\node  [main node] (l_j3) at (2,-0.5) {$l_{j,3}$};
\begin{scope}[-]
\draw [thin] (c_j) -- (a_j);
\draw [thin] (a_j) -- (b_j);
\draw [thin] (b_j) -- (l_j1);
\draw [thin] (b_j) -- (l_j2);
\draw [thin] (a_j) -- (l_j3);

\end{scope}

\end{tikzpicture}
    \subcaption{Clause gadget: $l_{j,k}=v_i$ or $l_{j,k}=\overline{v}_i$.}
\end{subfigure}
\qquad 
\begin{subfigure}[l]{.4\textwidth}\centering
\begin{tikzpicture}[->,
scale=1, every node/.style={anchor=center, scale=0.7},node distance=1cm,main node/.style={circle,fill=white!20,draw,font=\sffamily\Large\bfseries},u node/.style={rectangle,fill=white!20,draw,font=\sffamily\Large\bfseries},minimum height=22pt, minimum width = 22pt]
\node [u node] (w_i) at (0,0) {$w_i$};
\node [main node] (x_i) at (2,0.5) {$v_i$};
\node [main node] (y_i) at (1,0.7) {$y_i$};
\node [main node] (-y_i) at (1,-0.7) {$\overline{y}_i$};
\node [main node] (-x_i) at (2,-0.5) {$\overline{v}_i$};
\node [main node] (s) at (3,0) {$s$};
\node (nx_i) at (2.5,1) {};
\node  (n-x_i) at (2.5,-1) {};
\begin{scope}[-]
\draw [thin] (w_i) -- (-y_i);
\draw [thin] (w_i) -- (y_i);
\draw [thin] (x_i) -- (y_i);
\draw [thin] (-x_i) -- (-y_i);
\draw [thin] (s) -- (x_i);
\draw [thin] (-x_i) -- (x_i);
\draw [thin] (s) -- (-x_i);
\draw [thin] (x_i) -- (nx_i);
\draw [thin] (-x_i) -- (n-x_i);
\end{scope}

\end{tikzpicture}
    \subcaption{Variable selection gadget.}
\end{subfigure}

\caption{Gadgets for $G$: vertices from~$U$ are drawn as squares}
\label{fig:ext-cds:basic-block}
\end{figure}
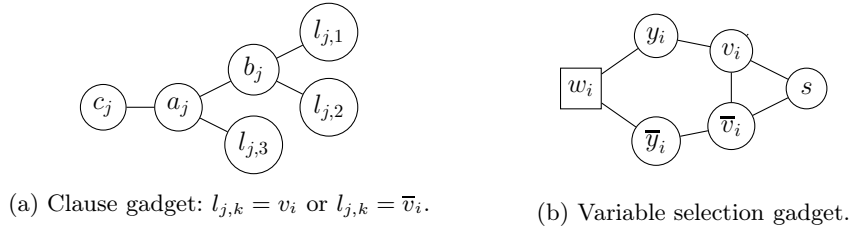

\begin{claim}
$(X,C)$ is a yes-instance of the \textsc{3-Sat} if and only if $(G,U)$ is a yes-instance of the \textsc{Connected Dominating Set Extension} problem.
\end{claim}

Let $D \subseteq V$ be a minimal CDS with $U \subseteq D$. Assume there exists an $i\in\{1,\ldots,n\}$ such that $v_i,\overline{v}_i\in D$. Since $D$ is connected and $w_i\in U \subseteq D$, $y_i$ or $\overline{y}_i$ are in $D$ (w.l.o.g., $y_i \in D$). Because $N[w_i]\setminus N[\{ v_i,\overline{v}_i,y_i\}] = \emptyset$ holds, $w_i$ has to be in $D$ for connectivity reasons. Therefore, $\overline{y}_i\in D$. Assume $p=(p_1,\ldots,p_l)$ is a walk with $p_k = w$ for $k\notin{1,n}$. Thus, $\{p_{k-1},p_k+1\}=\{ y_i,\overline{y}_i\}$. W.l.o.g., let $p_{k-1}= y_i$ and $p_{k+1}= \overline{y}_i$. Then also the walk $q=(p_1,\ldots, p_{k-1},v_i,\overline{v}_i,p_{k+1},\ldots,p_l)$ connects the same two vertices. Therefore, $D\setminus\{ w_i\}$ is also a CDS. Hence, for each $i\in \{1,\ldots,n\}$, there is either $v_i\in D$ or $\overline{v}_i\in D$ (at least one has to be in~$D$, otherwise $D$ is not connected because of~$w_i$). Define the interpretation $\phi= (\phi_1, \ldots, \phi_n)$ with 
$$\forall i\in\{1,\ldots,n\}:\, \phi_i = \begin{cases}
1, & v_i\in D\\
0, & \overline{v}_i \in D
\end{cases}.$$
 
Since for each $j\in \{1,\ldots,m\}$, $c_j$ has to be dominated and $N(c_j)=\{a_j\}$ holds, $a_j\in D$. Hence, there exists an $a_j$-$s$-path in $D \cup \{s\}$. Therefore, $\{t\}=N(a_j)\cap V'\subseteq D$ or $b_j\in D$. In the first case, $C_j$ is satisfied by $x_i$ if $t = v_i$ or $\neg x_i$ if $t= \overline{v}_i$ (for a $i\in\{1,\ldots,n\}$). If there exists $b_j\in D$, because of connectivity reasons, there exists  $i\in \{ 1,\ldots,n\}$ such that $v_i\in N(b_j)\cap D$ or $\overline{v}_i\in N(b_j)\cap D$. This implies that $C_j$ is satisfied by $x_i$ or $\neg x_i$. Thus, $\phi$ satisfies the given instance. Therefore, $(X,C)$ is a yes-instance.

Assume there exists an interpretation $\phi= (\phi_1,\ldots,\phi_n)$ that satisfies $(X,C)$. Then for each $j \in \{1,\ldots,m\}$, there exists  a $k_j \in \{1,2,3\}$ such that $l_{j,k_j}$ satisfies $C_j$. Define 
$$D= \left(\bigcup_{i=1}^n \{ w_i\} \cup V_i \right) \cup \{b_j \mid  l_{j,3}\text{ is false}, j\in \{1,\ldots,m\} \} \cup \{a_1, \ldots, a_m\}$$
with
$$V_i = \begin{cases} \{v_i,y_i\}, &\text{if }\phi_i=1\\
\{\overline{v}_i,\overline{y}_i\}, &\text{if } \phi_i=0 
\end{cases}$$ for each $i\in\{ 1, \ldots, n\}$. If $D$ is not connected (this can be checked in polynomial time), add $s$ to~$D$. Clearly, $U\subseteq D$. For each $i\in \{1,\ldots,n\}$ and $j\in \{1,\ldots,m\}$, $\{v_i,\overline{v}_i,w_i, y_i,\overline{y}_i\}\subseteq N[V_i]$, $\{a_j,b_j,c_j\}\subseteq N[a_j]$ hold. 
 This implies that $D$ is a dominating set.
 
 For each $i\in\{1, \ldots,n\}$ with $v_i\in D$, $(w_i, y_i, v_i, s)$ is a $w_i$-$s$-path (or $(w_i,\overline{y}_i, \overline{v}_i, s)$ if $\overline{v}_i\in D$) on $G[D\cup\{s\}]$. Furthermore, $(a_j,b_j,t_j, s)$ (if $l_{j,3}$ is false) or $(a_j,t_j, s)$ (if $l_{j,3}$ is true) is an $a_j$-$s$-path on $G[D\cup\{s\}]$ for $l_{j,k_j} \in \{x_{r_{k_j}}, \neg x_{r_{k_j}}\}$ and $\{t_j\}= D\cap \{v_{r_{k_j}},\overline{v}_{r_{k_j}}\}$. Since these paths include each element in~$D$ and end in~$s$, $D\cup \{s\}$ is connected. Therefore, $D\cup \{s\}$ is a CDS (and $D$ itself also if $D$ was connected anyways).
 
 Since for each $j\in\{ 1, \ldots, m\}$, $N(c_j)\cap D = \{a_j\}$ holds, $D\setminus\{a_j\}$ is not a dominating set. If $l_{j,k_j}$ is false, $a_j$ has degree~0 on $G[D\setminus \{b_j\}]$. Therefore, $b_j$ has to be in~$D$.
 Let $v_i\in D$ (or $\overline{v}_i$, respectively) the $N[\overline{y}_i]\cap D=w_i$ (or $N[y_i]\cap D=w_i$, respectively). Thus, $w_i$ has to be in~$D$. Since $w_i$ has degree~1 in $G[D]$, its neighbor $y_i$ (or $\overline{y}_i$, respectively) has also to be in~$D$ to form a CDS. On $G[D\setminus \{v_i\}]$ (or $G[D\setminus \{\overline{v}_i\}]$, respectively) $\{w_i, y_i\}$ (or $\{w_i, \overline{y}_i\}$, respectively) is an isolated edge. Thus, $D\setminus \{ v_i\}$ (or $D\setminus\{\overline{v_i}\}$, respectively) is not connected, since there is no $w_i$-$a_m$-path. As mentioned at the beginning, $D$ contains $s$ only if $D\setminus \{s\}$ is not connected. Hence, $D$ is a minimal CDS.\hfill$\Diamond$
 
 These two claims finish our argument.
\qed \end{proof}

\noindent
Due to the parsimonious nature of the reduction, we can also conclude the following result.

\begin{corollary}
Assuming that the Exponential Time Hypothesis\footnote{For a discussion of this hypothesis, we refer to \cite{ImpPatZan2001,LokMarSau2011b}.} holds, there is no algorithm that solves the \textsc{Connected Dominating Set Extension} problem in time $\Oh(2^{o(n)})$ on (2-degenerate) graphs of order~$n$.
\end{corollary}

Hence, even any subexponential delay seems to be hard to achieve.

Furthermore, the \textsc{Connected Dominating Set Extension} with the standard parameterization $\vert U\vert$, where $U$ is the given subset of $V$, is even \textsf{W}[3]-hard on split graphs. To show this, we need the \textsf{W}[3]-completeness of \textsc{Hitting Set Extension} \cite{BlaFLMS2019,BlaFLMS2022,BlaFriSch2022,CasFGMS2022} with standard  parameterization; in this problem, the input consists of a hypergraph $(X,\mathcal{S})$, with $\mathcal{S}\subseteq 2^X$, and a set $U\subseteq X$, and the question if there exists a minimal hitting set $H$ (this means that $H\cap S\neq\emptyset$ for all $S\in\mathcal{S}$) that extends $U$, \emph{i.e.}, $U\subseteq H$.  

\begin{theorem}\label{thm:Ext-CDS-Wthreehard}
\textsc{Connected Dominating Set Extension} (with standard parameterization) is \textsf{W}[3]-hard, even on split graphs.
\end{theorem}
\begin{proof}
Let $X=\lbrace x_1,\ldots,x_n\rbrace$ be a set and $U\coloneqq\lbrace u_1,\ldots, u_k\rbrace\subseteq X$, $S\coloneqq \lbrace S_1,\ldots,S_m\rbrace\subseteq 2^X$ such that $\emptyset$ is not a hitting set on $S$.
 Let $G=\left(V,E\right)$ be a graph with \begin{equation*}
\begin{split}
V\coloneqq& \lbrace s_1,\ldots,s_m\rbrace\cup X,  \\
E\coloneqq& \{ \{x_i,x_j\}\mid i,j\in \{1,\ldots,n\}, i<j\}\cup \bigcup_{i=1}^m \left(\lbrace \lbrace x_i,s_j\rbrace\vert x_i\in S_j\rbrace \right).
\end{split}
\end{equation*}
Clearly, $X$ is a clique on $G$ and for each $s_j\in S' = \{s_1,\ldots, s_m\}$, $N(s_j) \subseteq X$. Thus, $S'$ is a independent set of simplicial vertices and $G$ is a chordal graph (more particularly, a split graph).
\begin{claim}
$(X,S,U)$ is a yes-instance of  \textsc{Hitting Set Extension} if and only if $(G,U)$ is a yes-instance of  \textsc{Connected Dominating Set Extension}.
\end{claim}
Let $H = \{x_{h_1},\ldots,x_{h_l}\}\subseteq X$ be a minimal hitting set with $ U\subseteq H$. We show that $H$ is a minimal CDS. $H$ is connected in $G$, as $H\subseteq X$ is a clique. Since $H$ is a hitting set, for each $S_j\in S$ there exists some $x_i\in H\cap S_j$. Thus, for each $s_j \in S'$ is dominated by $x_i\in H\cap N(s_j)$. Furthermore, $H$ is not empty. Therefore, $X\cap H\neq \emptyset$ holds and $H$ is a CDS. Assume there exists a $x_{h_i}\in H$ such that $H\setminus \{x_{h_i}\}$ is also a CDS. As $H$ is a clique in $G$, deleting one vertex would not change the connectivity. This implies that for each $s_j\in N(x_{h_i})\cap S'$, there exists a $x_{d_j}\in (H\setminus \{x_{h_i}\})\cap N(s_j)$. Hence, for each $S_j\in S$ with $x_{h_i}\in S_j$, there exists a $x_{d_j}\in (H\setminus \{x_{h_i}\})\cap S_j$. This contradicts the minimality of $H$ (as a hitting set). Thus, $H$ is a minimal CDS in $G$ with $U\subseteq H$.

Let $D$ be a minimal CDS in $G$ with $U\subseteq D$. Assume there exists a $s_j\in D\cap S'$. If $D=\{s_j\}$, $S'=\{s_j\}$ holds, as $S'$ is an independent set. Hence, for each $x_i\in X$, $\{x_i\}$ is a CDS. Let $D\setminus \{s_j\}\neq\emptyset$.  Since $s_j$ is simplicial, for each $v\in N(s_j)\subseteq X$, $N[s_j]\subseteq N[v]$. As $D$ has elements other than $s_j$, $D\cap N(s_j)\neq \emptyset$ holds for connectivity reasons. This implies $D \setminus \{ s_j \}$ is also a CDS, as it is connected (every walk on $G[D]$ that include $s_j$ is also a walk $G[D\setminus \{s_j\}]$ if you remove $s_j$, since its neighbors are neighbors themselves) and is a dominating set.
Therefore, we assume $D = \{x_{d_1}, \ldots, x_{d_l}\} \subseteq X$. Let $S_j\in S$. Since $s_j$ is dominated by a $x_i\in D$, there is $x_i\in S_j \cap D$. Thus, $D$ is a hitting set. Assume there exists a $x_{d_i}\in D$ such that $D\setminus \{ x_{d_i}\}$ is a hitting set. This implies that for each $S_j\in S$ with $x_{d_i}\in S_j$, there exists a $x_{z_j}\in (D \setminus \{x_{d_i}\}) \cap S_j$. Thus, for all $s_j\in N(x_{d_i})$, there is some $x_{z_j}\in (D\setminus \{x_{d_i}\})\cap N(s_j)$. This contradicts the minimality of $D$:  $x_{d_i}$ cannot be in $D$ for connectivity reasons, as $D$ is a clique. Hence, $D$ is a minimal hitting set.
\qed \end{proof}

\begin{remark}
For split graphs, also the reverse of the reduction from \autoref{thm:Ext-CDS-Wthreehard} is true: a minimal CDS in a split graph can be interpreted as a minimal hitting set in the naturally related \textsc{Hitting Set} instance; more details below, also view \autoref{fig:ext-cds:basic-block1}. Therefore, we even have \textsf{W}[3]-completeness of \textsc{Connected Dominating Set Extension} (with standard parameterization) on split graphs. For general graphs, \textsf{W}[3]-membership is still open.
\end{remark}

\begin{theorem}\label{thm:Ext-CDS-Wthreecomplete-split}
\textsc{Connected Dominating Set Extension} (with standard parameterization) is \textsf{W}[3]-complete on split graphs.
\end{theorem}
\begin{proof}
For the completeness, we will again use the \textsc{Hitting Set Extension} with standard parameterization.

Let $G=(V,E)$ a split graph such that there is a clique~$C$ and an independent set $I$ with $V=I \cup C$. By the argumentation from the proof of \autoref{thm:Ext-CDS-Wthreehard}, each minimal CDS~$D$ on~$G$ fulfills $D\cap I=\emptyset.$  

\begin{claim}
$(G,U)$ with $U\subseteq V$ a yes-instance of \textsc{Connected Dominating Set Extension} if and only if $(C,\{ N(i)\mid i\in I\},U)$ is a yes-instance of the \textsc{Hitting Set Extension}.
\end{claim}
As mentioned before, if $U\cap I\neq \emptyset$ holds, this is a no-instance. Therefore, we assume  $U\cap I= \emptyset$. The proof of this claim is the same as the proof of the claim in \autoref{thm:Ext-CDS-Wthreehard}, since the reduction (of \autoref{thm:Ext-CDS-Wthreehard}) would lead to a graph equivalent to~$G$, if we apply it to $(C,\{ N(i)\mid i\in I\},U)$ ($U'$ would be the the image of $U$ by the reduction). 
\qed \end{proof}

\begin{figure}[tb]
\centering
\begin{subfigure}[l]{.5\textwidth}\centering
    \begin{tikzpicture}[->,
scale=1, every node/.style={anchor=center, scale=0.7},node distance=1cm,main node/.style={circle,fill=white!20,draw,font=\sffamily\Large\bfseries}]
\node [main node] (a_1) at (0,0) {$x_1$};
\node [main node] (a_2) at (1,-1) {$x_2$};
\node [main node] (a_3) at (-1,-1) {$x_3$};
\node  [main node] (t_1) at (0,1) {\color{blue}$s_1$};
\node  [main node] (t_2) at (1.5,0) {\color{orange}$s_2$};
\node  [main node] (t_3) at (0,-2) {\color{Green}$s_3$};
\node  [main node] (t_4) at (-2,-2) {\color{Brown}$s_4$};
\begin{scope}[-]
\draw [thin] (a_1) -- (a_2);
\draw [thin] (a_2) -- (a_3);
\draw [thin] (a_3) -- (a_1);
\draw [thin] (t_1) -- (a_1);
\draw [thin] (t_1) -- (a_2);
\draw [thin] (t_1) -- (a_3);
\draw [thin] (t_2) -- (a_1);
\draw [thin] (t_2) -- (a_2);
\draw [thin] (t_3) -- (a_2);
\draw [thin] (t_3) -- (a_3);
\draw [thin] (t_4) -- (a_3);
\end{scope}

\end{tikzpicture}
    \subcaption{The original CDS instance, a split graph.}
\end{subfigure}
\qquad 
\begin{subfigure}[l]{.41\textwidth}\centering
\begin{tikzpicture}[->,
scale=1, every node/.style={anchor=center, scale=0.7},node distance=2cm,main node/.style={circle,fill=white!20,draw,font=\sffamily\Large\bfseries},u node/.style={rectangle,fill=white!20,draw,font=\sffamily\Large\bfseries},minimum height=25pt, minimum width = 25pt]
\node [circle,draw] (x_1) at (0,1.2) {$x_1$};
\node [circle,draw] (x_2) at (1,0) {$x_2$};
\node [circle,draw] (x_3) at (-1,0) {$x_3$};
\node [ellipse,draw,fit=(x_1)(x_2)(x_3),color=blue,minimum height=12em, minimum width = 16em]{};
\node [ellipse,draw,fit=(x_2)(x_3),color=Green,minimum height=6em, minimum width = 13em]{};
\node [ellipse,draw,fit=(x_2)(x_1),color=orange,minimum height=2em, minimum width = 10em,rotate=-45]{};
\node [ellipse,draw,fit=(x_3),color=Brown,minimum height=3.8em, minimum width = 3.8em]{};
\end{tikzpicture}
    \subcaption{The equivalent \textsc{Hitting Set} instance: mind the colors that define {\color{blue}$s_1$}, {\color{orange}$s_2$}, {\color{Green}$s_3$} and~{\color{Brown}$s_4$}.}
\end{subfigure}

\caption{Reduction between \textsc{Connected Dominating Set Extension} and \textsc{Hitting Set Extension}; the set $U$ is not specified.}
\label{fig:ext-cds:basic-block1}
\end{figure}
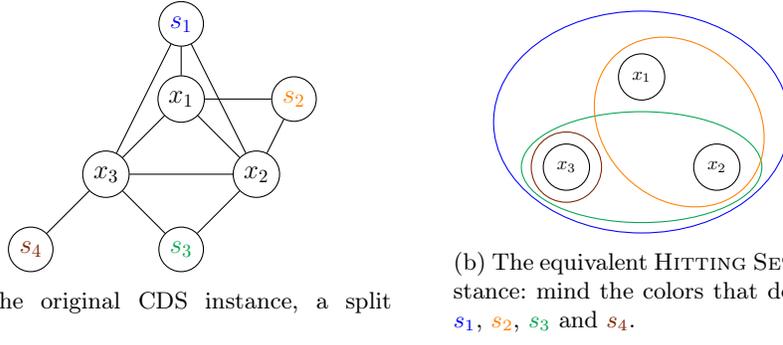

There is yet another nearly final stroke against hopes for obtaining a polynomial-delay enumeration algorithm for minimal CDS. Namely, we could see the reduction presented for \autoref{thm:Ext-CDS-Wthreehard} also as a reduction that proves the following statement, which connects our enumeration problem to \textsc{Hitting Set Transversal}, a problem notoriously open for decades.

\begin{theorem}
If there was an algorithm for enumerating minimal CDS with polynomial delay in split graphs, then there would be an algorithm for enumerating minimal hitting sets in hypergraphs with polynomial delay.
\end{theorem}

\section{Lower bounds}\label{sec-lowerbounds}
 
Several attempts to construct lower bound examples are known from the literature, leading to $3^{(n-2)/3}\in\Omega(\oldcdslowerbound^n)$ many minimal connected dominating sets in $n$-vertex graphs~\cite{GolHegKra2016,Say2019}. The construction of Sayadi~\cite{Say2019} is interesting insofar, as the resulting graphs are not only convex bipartite, but also 3-degenerate.

We now present an improved lower bound on the maximum number of minimal connected dominating sets in a graph. Further details about this construction can be found in \cite{Abu2021}.
Given arbitrary positive integers $k,t$, we construct a graph $G_t^k$ of order $n = k(2t+1) + 1$ as follows. The main building blocks of $G_t^k$ consist of $k$ copies of a base-graph $G_{t}$, of order $2t-1$. The vertex set of $G_{t}$ consists of three layers. The first one is a set  $X=\{x_{1},\ldots, x_{t}\}$ that induces a clique. The second one is an independent set $Y=\{y_{1},\ldots, y_{t}\}$, while the third layer consists of a singleton $\{z\}$.
Each vertex $x_{i}\in X$ has exactly $t-1$ neighbors in $Y$: $N(x_{i}) = \{y_{j}\in Y: i\neq j\}$. Hence, $X\cup Y$ induces a copy of $K_{t,t}$ minus a perfect matching. Finally the vertex $z$ is adjacent to all vertices in $Y$. Figure \ref{basic-block} shows the graph $G_t$ for $t=4$.
To finally construct the graph $G_t^k$, we introduce a final vertex~$s$ that is connected to all vertices of each set~$X$ of each copy of the base-graph.

\begin{lemma}
\label{Gi}
For each $t>0$, the graph $G_t$ has exactly $\frac{t^3+t^2}{2}-t$ minimal connected dominating sets that have non-empty intersection with the set $X$.
\end{lemma}

\begin{proof}
The set $X$ cannot have more than two vertices in common with any minimal CDS, since any two elements of~$X$ dominate $X\cup Y$. Any minimal CDS that contains exactly one vertex~$x_i$ of~$X$ must contain the vertex $z$, to dominate $y_i$, and one of the $t-1$ neighbors of~$x_i$ (to be connected). There are $t(t-1)$ sets of this type.
Moreover, each pair of elements of~$X$ dominates~$Y$. So a minimal CDS can be formed by (any) two elements of $X$ and any of the elements of $Y$ (to dominate $z$). There are $t\frac{t(t-1)}{2}$ such sets. 
\qed \end{proof}

The hub-vertex $s$ in $G_t^k$ must be in any CDS, being a cut-vertex. Therefore, there is no need for the set~$X$ in $G_t$ to induce a clique, being always dominated by~$s$. In other words, the counting used in the proof above still holds if each copy of $G_t$ is replaced by $G_{t}-E(X)$ in $G_t^k$. Here, $E(X)$ denotes the set of edges in $G_t[X]$.
\autoref{G3-3} shows $G_3^3$ without the edges between pairs of element of~$X$ in each copy of~$G_3$. With the help of \autoref{Gi}, we can show:

\begin{figure}[tb]
\centering
\begin{subfigure}[b]{.28\textwidth}
\centering
\scalebox{.78}{\begin{tikzpicture}[->,
scale=0.2, every node/.style={anchor=center, scale=0.65},node distance=1cm,main node/.style={circle,fill=white!20,draw,font=\sffamily\Large\bfseries}]

\node [main node] (1) at (-2,0) {\(~z~\)};
\node[main node] (2) at (-12,-4) {\(y_1\)};
\node [main node] (3) at (-6,-4) {\(y_2\)};
\node [main node] (4) at (2,-4) {\(y_3\)};
\node [main node] (5) at (8,-4) {\(y_3\)};
\node[main node] (6) at (-12,-12) {\(x_1\)};
\node [main node] (7) at (-6,-12) {\(x_2\)};
\node [main node] (8) at (2,-12) {\(x_3\)};
\node [main node] (9) at (8,-12) {\(x_4\)};

\begin{scope}[-]

\draw [thin] (1) -- (2);
\draw [thin] (1) -- (3);
\draw [thin] (1) -- (4);
\draw [thin] (1) -- (5);
\draw [thin] (2) -- (7);
\draw [thin] (2) -- (8);
\draw [thin] (2) -- (9);
\draw [thin] (3) -- (6);
\draw [thin] (3) -- (8);
\draw [thin] (3) -- (9);
\draw [thin] (4) -- (6);
\draw [thin] (4) -- (7);
\draw [thin] (4) -- (9);
\draw [thin] (5) -- (6);
\draw [thin] (5) -- (7);
\draw [thin] (5) -- (8);
\draw [thin] (6) -- (7);
\draw [thin] (6) to[bend right] (8);
\draw [thin] (6) to[bend right] (9);
\draw [thin] (7) -- (8);
\draw [thin] (8) -- (9);
\draw [thin] (7) to[bend right] (9);

\end{scope}

\end{tikzpicture}}

\subcaption{The graph $G_4$ }
\label{basic-block}
\end{subfigure}
\qquad 
\begin{subfigure}[b]{.6\textwidth}
\centering
\scalebox{0.78}{\begin{tikzpicture}[->,
scale=0.2, every node/.style={anchor=center, scale=0.6},node distance=1cm,main node/.style={circle,fill=white!20,draw,font=\sffamily\Large\bfseries}]

\node [main node] (1) at (-17,0) {\(~z_1~\)};
\node[main node] (2) at (-22,-4) {\(y_{12}\)};
\node [main node] (3) at (-17,-4) {\(y_{13}\)};
\node [main node] (4) at (-12,-4) {\(y_{11}\)};
\node[main node] (5) at (-22,-8) {\(x_{11}\)};
\node [main node] (6) at (-17,-8) {\(x_{12}\)};
\node [main node] (7) at (-12,-8) {\(x_{13}\)};

\node [main node] (8) at (-2,0) {\(~z_2~\)};
\node[main node] (9) at (-7,-4) {\(y_{22}\)};
\node [main node] (10) at (-2,-4) {\(y_{23}\)};
\node [main node] (11) at (3,-4) {\(y_{21}\)};
\node[main node] (12) at (-7,-8) {\(x_{21}\)};
\node [main node] (13) at (-2,-8) {\(x_{22}\)};
\node [main node] (14) at (3,-8) {\(x_{23}\)};

\node [main node] (15) at (13,0) {\(~z_3~\)};
\node[main node] (16) at (8,-4) {\(y_{32}\)};
\node [main node] (17) at (13,-4) {\(y_{33}\)};
\node [main node] (18) at (18,-4) {\(y_{31}\)};
\node[main node] (19) at (8,-8) {\(x_{31}\)};
\node [main node] (20) at (13,-8) {\(x_{32}\)};
\node [main node] (21) at (18,-8) {\(x_{33}\)};

\node [main node] (22) at (-2,-13) {\(~s~\)};

\node[] (50) at (-17,2) {};

\begin{scope}[-]

\draw [thin] (8) -- (9);
\draw [thin] (8) -- (10);
\draw [thin] (8) -- (11);
\draw [thin] (12) -- (9);
\draw [thin] (12) -- (10);
\draw [thin] (13) -- (10);
\draw [thin] (13) -- (11);
\draw [thin] (14) -- (11);

\draw [thin] (22) to[bend left] (5);
\draw [thin] (22) to[bend left] (6);
\draw [thin] (22) -- (7);
\draw [thin] (22) -- (12);
\draw [thin] (22) -- (13);
\draw [thin] (22) -- (14);
\draw [thin] (22) -- (19);
\draw [thin] (22) to[bend right] (20);
\draw [thin] (22) to[bend right] (21);

\draw [thin] (1) -- (2);
\draw [thin] (1) -- (3);
\draw [thin] (1) -- (4);
\draw [thin] (5) -- (2);
\draw [thin] (5) -- (3);
\draw [thin] (6) -- (3);
\draw [thin] (6) -- (4);
\draw [thin] (7) -- (4);
\draw (2) .. controls (-17,8) and (-7,-2) .. (7);
\draw (9) .. controls (-2,8) and (8,-2) .. (14);

\draw (16) .. controls (13,8) and (23,-2) .. (21);

\draw [thin] (15) -- (16);
\draw [thin] (15) -- (17);
\draw [thin] (15) -- (18);
\draw [thin] (19) -- (16);
\draw [thin] (19) -- (17);
\draw [thin] (20) -- (17);
\draw [thin] (20) -- (18);
\draw [thin] (21) -- (18);

\end{scope}

\end{tikzpicture}}

\subcaption{The graph $G_3^3$ }
\label{G3-3}
\end{subfigure}
\caption{How our lower bound examples are composed.}
\end{figure}
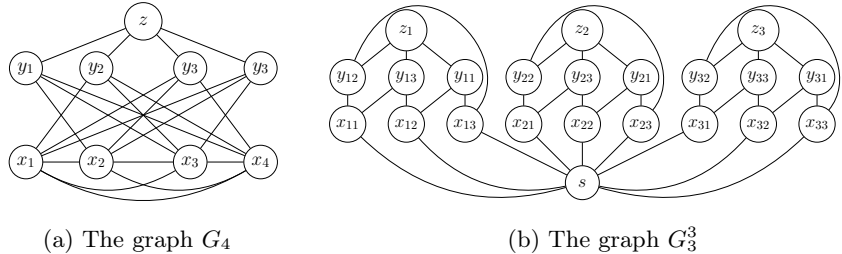

\begin{theorem}
\label{lowerbound}
The maximum number of minimal CDS in a  connected graph of order~$n$ is in $\Omega(\newcdslowerbound^n)$; an example family of graphs is $(G_4^k)$.
\end{theorem}

\begin{proof}
By \autoref{Gi}, each copy of the graph $G_{t}$ has $\frac{t^3+t^2}{2}-t$ minimal CDS. There are $k$ such graphs in $G_t^k$, in addition to the vertex~$s$ that connects them all. 
Every minimal CDS must contain~$s$ and at least one element from $N(s)$ in each $G_{t}$.
Therefore, the total number of minimal CDS in $G_t^k$ is $(\frac{t^3+t^2}{2}-t)^k = (\frac{t^3+t^2}{2}-t)^{\frac{n-1}{2t+1}}$. The claimed lower bound is achieved when $t=4$, which gives a total of $36^{\frac{n-1}{9}} \in \Omega(\newcdslowerbound^n)$.
\qed \end{proof}

We note that $G_t^k$ is a $t$-degenerate graph that is also bipartite (since the set $X$ in each copy of $G_{t}$ can be an independent set). Furthermore, $G_3^k$ is planar, as intentionally drawn in  \autoref{G3-3}. Our general formula (\autoref{Gi}) yields that $G_3^k$ 
has $15^{n/7}=\Omega(\newcdsthreedegenlowerbound^n)$ minimal CDS, hence improving on the previously mentioned construction for 3-degenerate graphs in~\cite{Say2019}. 

\begin{corollary}
The maximum number of minimal connected dominating sets in a 3-degenerate bipartite planar graph of order~$n$ is in $\Omega(\newcdsthreedegenlowerbound^n)$.
\end{corollary}

Finally, $G_2^k$ is a 2-degenerate graph of order $n$ with $4^{n/5}=\Omega(\newcdstwodegenlowerbound^n)$ many minimal CDS, incidentally matching the best known lower bound 
in cobipartite graphs \cite{CouHHK2012}.

\section{Conclusions and Open Problems}

In our paper, we focused on developing an input-sensitive enumeration algorithm for minimal CDS. We achieved some notable progress both on the running time for such enumeration algorithms and with respect to the lower bound examples. However, the gap between lower and upper bound is still quite big, and the natural question to ask here is to bring lower and upper bounds closer; in an optimal setting, both would match. We are working on a further refined version that will bring down the upper bound a bit, but not decisively.
This question of non-matching upper and lower bounds is also open for most special graph classes. 

One particular such graph class that is studied in this paper is the class of 2-degenerate graphs. We like to suggest to study this graph class also for other enumeration problems, or, more generally, for problems that involve a measure-and-conquer analysis of branching algorithms, because this was the key to break the $2^n$-barrier significantly for enumerating minimal CDS with measure-and-conquer, something that seemed to be impossible with other more standard approaches, like putting weights to low-degree vertices.

As we also proved that the extension problem associated to CDS is computationally intractable even on 2-degenerate graphs, it is not that straightforward to analyze our enumeration algorithm with the eyes of output-sensitive analysis. Conversely, should it be possible to find an efficient algorithm for an extension problem, also on special graph classes, then usually polynomial-delay algorithms can be shown; as a recent example in the realm of domination problems, we refer to the enumeration of minimal Roman  functions described in~\cite{AbuFerMan2022}. So, in the context of our problem, we can ask: 
Can we achieve polynomial delay for any enumeration algorithm for minimal CDS? Can we combine this analysis with a good input-sensitive enumeration approach?
Notice that the corresponding questions for enumerating minimal dominating sets are an open question for decades. This is also known as the \textsc{Hitting Set Transversal}  problem; see \cite{CreKPSV2019,EitGot95,GaiVer2017,KanLMN2014}. We therefore also presented relations between polynomial-delay enumeration of minimal dominating sets and that of minimal CDS, again explaining the difficulty of the latter question. Finally, we also briefly discussed the possibility of subexponential delay. We propose to discuss this question further also for other enumeration problems when polynomial delay is not achievable, as it might well be a practical solution to know that the delay time is substantially smaller than the time needed to enumerate all solutions, but not polynomial time. 

We also discussed parameterized complexity aspects of \textsc{Connected Dominating Set Extension}, leaving \textsf{W}[3]-membership an an open question in the general case.


\begin{thebibliography}{10}

\bibitem{Abu2021}
F.~N. Abu{-}Khzam.
\newblock A note on the maximum number of minimal connected dominating sets in
  a graph.
\newblock Technical Report 2111.06026, Cornell University, ArXiv/CoRR, 2021.
\newblock URL: \url{https://arxiv.org/abs/2111.06026}.

\bibitem{AbuFerMan2022}
F.~N. Abu{-}Khzam, H.~Fernau, and K.~Mann.
\newblock Minimal {R}oman dominating functions: Extensions and enumeration.
\newblock Technical Report 2204.04765, Cornell University, ArXiv/CoRR, 2022.
\newblock URL: \url{https://doi.org/10.48550/arXiv.2204.04765}.

\bibitem{AbuMouLie2011}
F.~N. Abu{-}Khzam, A.~E. Mouawad, and M.~Liedloff.
\newblock An exact algorithm for connected red-blue dominating set.
\newblock {\em Journal of Discrete Algorithms}, 9(3):252--262, 2011.

\bibitem{BlaFLMS2019}
T.~Bl\"asius, T.~Friedrich, J.~Lischeid, K.~Meeks, and M.~Schirneck.
\newblock Efficiently enumerating hitting sets of hypergraphs arising in data
  profiling.
\newblock In {\em Algorithm Engineering and Experiments (ALENEX)}, pages
  130--143. {SIAM}, 2019.

\bibitem{BlaFLMS2022}
T.~Bl{\"{a}}sius, T.~Friedrich, J.~Lischeid, K.~Meeks, and M.~Schirneck.
\newblock Efficiently enumerating hitting sets of hypergraphs arising in data
  profiling.
\newblock {\em Journal of Computer and System Sciences}, 124:192--213, 2022.

\bibitem{BlaFriSch2022}
T.~Bl{\"{a}}sius, T.~Friedrich, and M.~Schirneck.
\newblock The complexity of dependency detection and discovery in relational
  databases.
\newblock {\em Theoretical Computer Science}, 900:79--96, 2022.

\bibitem{Utrecht-TR-2015-016}
H.~L. Bodlaender, E.~Boros, P.~Heggernes, and D.~Kratsch.
\newblock Open problems of the {L}orentz workshop ``{E}numeration algorithms
  using structure''.
\newblock Technical Report UU-CS-2015-016, Department of Information and
  Computing Sciences, Utrecht University, Utrecht, The Netherlands, November
  2015.

\bibitem{CasFGMS2022}
K.~Casel, H.~Fernau, M.~Khosravian Ghadikolaei, J.~Monnot, and F.~Sikora.
\newblock On the complexity of solution extension of optimization problems.
\newblock {\em Theoretical Computer Science}, 904:48--65, 2022.
\newblock \href {https://doi.org/https://doi.org/10.1016/j.tcs.2021.10.017}
  {\path{doi:https://doi.org/10.1016/j.tcs.2021.10.017}}.

\bibitem{CouHHK2012}
J.{-}F. Couturier, P.~Heggernes, P.~van~'t Hof, and D.~Kratsch.
\newblock Minimal dominating sets in graph classes: Combinatorial bounds and
  enumeration.
\newblock In M.~Bielikov{\'{a}}, G.~Friedrich, G.~Gottlob, S.~Katzenbeisser,
  and G.~Tur{\'{a}}n, editors, {\em {SOFSEM} 2012: Theory and Practice of
  Computer Science - 38th Conference on Current Trends in Theory and Practice
  of Computer Science}, volume 7147 of {\em LNCS}, pages 202--213. Springer,
  2012.

\bibitem{CouLetLie2015}
J.{-}F. Couturier, R.~Letourneur, and M.~Liedloff.
\newblock On the number of minimal dominating sets on some graph classes.
\newblock {\em Theoretical Computer Science}, 562:634--642, 2015.

\bibitem{CreKPSV2019}
N.~Creignou, M.~Kr{\"{o}}ll, R.~Pichler, S.~Skritek, and H.~Vollmer.
\newblock A complexity theory for hard enumeration problems.
\newblock {\em Discrete Applied Mathematics}, 268:191--209, 2019.

\bibitem{EitGot95}
T.~Eiter and G.~Gottlob.
\newblock Identifying the minimal transversals of a hypergraph and related
  problems.
\newblock {\em {SIAM} Journal on Computing}, 24(6):1278--1304, 1995.

\bibitem{FerGolSag2018}
H.~Fernau, P.~A. Golovach, and M.{-}F. Sagot.
\newblock Algorithmic enumeration: Output-sensitive, input-sensitive,
  parameterized, approximative ({D}agstuhl {S}eminar 18421).
\newblock {\em Dagstuhl Reports}, 8(10):63--86, 2018.

\bibitem{FomGraKra2008}
F.~V. Fomin, F.~Grandoni, and D.~Kratsch.
\newblock Solving {C}onnected {D}ominating {S}et faster than $2^n$.
\newblock {\em Algorithmica}, 52:153--166, 2008.

\bibitem{FomGraKra2009}
F.~V. Fomin, F.~Grandoni, and D.~Kratsch.
\newblock A measure \& conquer approach for the analysis of exact algorithms.
\newblock {\em Journal of the ACM}, 56(5), 2009.

\bibitem{Fometal2008a}
F.~V. Fomin, F.~Grandoni, A.~V. Pyatkin, and A.~A. Stepanov.
\newblock Combinatorial bounds via measure and conquer: Bounding minimal
  dominating sets and applications.
\newblock {\em ACM Trans. Algorithms}, 5(1):1--17, 2008.

\bibitem{FomKra2010}
F.~V. Fomin and D.~Kratsch.
\newblock {\em Exact Exponential Algorithms}.
\newblock Texts in Theoretical Computer Science. Springer, 2010.

\bibitem{GaiVer2017}
A.~Gainer{-}Dewar and P.~Vera{-}Licona.
\newblock The minimal hitting set generation problem: Algorithms and
  computation.
\newblock {\em {SIAM} Journal of Discrete Mathematics}, 31(1):63--100, 2017.

\bibitem{GolHKKV2017}
P.~A. Golovach, P.~Heggernes, M.~Moustapha Kant{\'{e}}, D.~Kratsch, and
  Y.~Villanger.
\newblock Minimal dominating sets in interval graphs and trees.
\newblock {\em Discrete Applied Mathematics}, 216:162--170, 2017.

\bibitem{GolHegKra2016}
P.~A. Golovach, P.~Heggernes, and D.~Kratsch.
\newblock Enumerating minimal connected dominating sets in graphs of bounded
  chordality.
\newblock {\em Theoretical Computer Science}, 630:63--75, 2016.

\bibitem{GolHKS2020}
P.~A. Golovach, P.~Heggernes, D.~Kratsch, and R.~Saei.
\newblock Enumeration of minimal connected dominating sets for chordal graphs.
\newblock {\em Discrete Applied Mathematics}, 278:3--11, 2020.

\bibitem{HHS98}
T.~W. Haynes, S.~T. Hedetniemi, and P.~J. Slater.
\newblock {\em Fundamentals of Domination in Graphs}, volume 208 of {\em
  Monographs and Textbooks in Pure and Applied Mathematics}.
\newblock Marcel Dekker, 1998.

\bibitem{ImpPatZan2001}
R.~Impagliazzo, R.~Paturi, and F.~Zane.
\newblock Which problems have strongly exponential complexity?
\newblock {\em Journal of Computer and System Sciences}, 63(4):512--530, 2001.

\bibitem{Iwa1112}
Y.~Iwata.
\newblock A faster algorithm for dominating set analyzed by the potential
  method.
\newblock In D.~Marx and P.~Rossmanith, editors, {\em Parameterized and Exact
  Computation - 6th International Symposium, {IPEC} 2011}, volume 7112 of {\em
  LNCS}, pages 41--54. Springer, 2012.

\bibitem{KanKKK2018}
K.~Kangas, P.~Kaski, J.~H. Korhonen, and M.~Koivisto.
\newblock On the number of connected sets in bounded degree graphs.
\newblock {\em Electron. J. Comb.}, 25(4):P4.34, 2018.

\bibitem{KanLMN2014}
M.~M. Kant{\'e}, V.~Limouzy, A.~Mary, and L.~Nourine.
\newblock On the enumeration of minimal dominating sets and related notions.
\newblock {\em {SIAM} Journal of Discrete Mathematics}, 28(4):1916--1929, 2014.

\bibitem{Law76}
E.~L. Lawler.
\newblock A note on the complexity of the chromatic number problem.
\newblock {\em Information Processing Letters}, 5(3):66--67, 1976.

\bibitem{LokMarSau2011b}
D.~Lokshtanov, D.~Marx, and S.~Saurabh.
\newblock Lower bounds based on the {E}xponential {T}ime {H}ypothesis.
\newblock {\em EATCS Bulletin}, 105:41--72, 2011.

\bibitem{LokPilSau2018}
D.~Lokshtanov, M.~Pilipczuk, and S.~Saurabh.
\newblock Below all subsets for minimal connected dominating set.
\newblock {\em {SIAM} Journal of Discrete Mathematics}, 32(3):2332--2345, 2018.

\bibitem{MooMos65}
J.~W. Moon and L.~Moser.
\newblock On cliques in graphs.
\newblock {\em Israel Journal of Mathematics}, 3:23--28, 1965.

\bibitem{NedRooDij2014}
J.~Nederlof, J.~M.~M. van Rooij, and T.~C. van Dijk.
\newblock Inclusion/exclusion meets measure and conquer.
\newblock {\em Algorithmica}, 69(3):685--740, 2014.

\bibitem{Say2019a}
M.~Y. Sayadi.
\newblock {\em Construction et analyse des algorithmes exacts et exponentiels:
  {\'{e}}num{\'{e}}ration input-sensitive. ({D}esign and analysis of exact
  exponential algorithms: input-sensitive enumeration)}.
\newblock PhD thesis, University of Lorraine, Nancy, France, 2019.

\bibitem{Say2019}
M.~Y. Sayadi.
\newblock On the maximum number of minimal connected dominating sets in convex
  bipartite graphs.
\newblock Technical Report abs/1908.02174, Cornell University, ArXiv, 2019.

\bibitem{Skj2017}
I.~B. Skj{\o}rten.
\newblock Faster enumeration of minimal connected dominating sets in split
  graphs.
\newblock Master's thesis, Department of Informatics, University of Bergen,
  Norway, June 2017.

\end{thebibliography}

\end{document}